\documentclass[final]{siamltex}

\usepackage{multirow}
\usepackage{setspace}
\usepackage{amssymb,amsmath,color,graphicx,subfigure}
\usepackage[noadjust]{cite}
\usepackage{enumerate,amsmath,epsfig}
\usepackage{graphics}
\usepackage{appendix}

\title{Biosensor Arrays for Estimating Molecular Concentration in Fluid Flows}

\author{Maryam Abolfath-Beygi 
        \and Vikram Krishnamurthy\thanks{Department of Electrical and Computer Engineering, University of British Columbia, Vancouver, Canada ({\tt \{maryabd,vikramk\}@ece.ubc.ca}).}}
        
\begin{document}

\maketitle

\begin{abstract}
This paper constructs dynamical models and estimation algorithms for the  concentration of target molecules in a fluid flow using an array of novel biosensors. Each biosensor is constructed out of protein molecules  embedded in a synthetic cell membrane. The concentration evolves according to an advection-diffusion partial differential equation which is coupled with chemical reaction equations on the biosensor surface. By using averaging theory methods and the divergence theorem, an approximate model is constructed that describes the asymptotic behaviour of the concentration as a system of ordinary differential equations. The estimate of target molecules is then obtained by solving a nonlinear least squares problem. It is shown that the estimator is strongly consistent and asymptotically normal. An explicit expression is obtained for the asymptotic variance of the estimation error. As an example, the results are illustrated for a novel biosensor built out of protein molecules.
\end{abstract}



\pagestyle{myheadings}
\thispagestyle{plain}
\markboth{MARYAM ABOLFATH-BEYGI AND VIKRAM KRISHNAMURTHY}{BIOSENSOR ARRAYS FOR ESTIMATING MOLECULAR CONCENTRATION IN FLUID FLOWS}

\section{Introduction}
Estimating the concentration of target molecules in a fluid flow over multiple biosensors is a challenging problem due to two non-standard features. Firstly, it is a parameter estimation problem of an advection-diffusion partial differential equation (PDE) which is coupled with Dirichlet and Neumann boundary conditions and cannot be solved analytically. Secondly, the measurement process affects the system state since each biosensor grabs target molecules and changes the concentration in the flow. This is  unusual since in most statistical inference problems, observation does not change the system state. The main results of this paper are briefly stated as follows:
\begin{remunerate}
\item An advection-diffusion PDE model is constructed to model the variations of the concentration of target molecules that flow past a linear array of biosensors. To facilitate estimation of the concentration of target molecules, Theorem~\ref{multiple-comp} develops an approximation method to describe the dynamics of the problem by a system of ordinary differential equations (ODEs). In Theorem~\ref{multiple-comp}, the multi-compartment ODE model is derived by exploiting the multiple time-scale behaviour of the system, together with perturbation methods and the divergence theorem.
\item  A novel biosensor constructed out of protein molecules is used as an actual example to illustrate our results. The development of this biosensor was first published in \emph{Nature} \cite{bib1}. The biosensor incorporates ion channels into a tethered lipid bilayer membrane where the
conductance of the channels is
switched by the recognition event.
 The PDE model for this biosensor is specified and solved numerically using the Comsol multi-physics finite element analysis software. We show how the PDE model and ODE approximations can satisfactorily model this novel biosensor. 
\item The estimation of target molecule concentration is posed as a parameter estimation problem in terms of the derived ODE model. The estimate is computed numerically for the novel biosensor via nonlinear least squares method.
The asymptotic behaviour of the estimator is analyzed. It is shown that the estimator is asymptotically unbiased and normal and its asymptotic variance is derived. According to the expression for the variance, the achievable  improvement in the estimate based on the number of biosensors is evaluated using the results from the ODE model.
\end{remunerate}
Inverse problems in fluid mechanics deal with estimating an unknown coefficient or function in the initial or boundary condition for a parabolic PDE (see \cite{banks,khoshgel,V,VI,cylinder}). In our case, the problem is estimating the boundary condition of a parabolic PDE.
The model-based state or parameter estimation of distributed systems (infinite dimensional systems)
based on a distributed-parameter description is quite complex.
 To address this problem, the system description is converted from
a distributed-parameter into a lumped-parameter form. This
conversion can be achieved by methods for solving partial
differential equations, such as finite-difference method \cite{16},
the finite-element method, modal analysis \cite{3} and spectral
method \cite{17}. 
 
The estimation problem in this paper can also be viewed as a source determination in a fluid system. In \cite{Wilson}, the strength (emission rate) of a contaminant source is estimated using a fixed network of concentration measurements and a Lagrangian trajectory model. In \cite{Flesch}, a backward-time Lagrangian stochastic model is used to estimate the emission rate of a surface area source. 
In \cite{Bayesian}, Baysian inference is applied to solve a chemical source determination problem where the posterior joint distribution of location, intensity, and temporal properties of a point source is obtained by a Markov chain Monte Carlo (MCMC) method. The authors develop a dual problem for the advection-diffusion equation using adjoint dispersion equations which
requires  significantly less amount of computations comparing to resolving the main equation for every source term. 

In the above works, the PDE has standard non-coupled initial and boundary conditions where measurements do not affect the system state whereas our problem has unconventional boundary conditions. 
%
In order to convert the system description from a distributed-parameter into a lumped-parameter form, an approximation method is used which is based on the two-compartment model \cite{cdc}. The two-compartment model is used in modeling a variety of binding experiments influenced by diffusion and mass transport. For example, in  \cite{Computational} and \cite{Extending}, this model is used to study and characterize the kinetic
properties of biomolecular interactions in  optical biosensors.
 In this work, the two-compartment model is extended and developed to a new multi-compartment model to describe the reaction-diffusion experiment in a flow chamber over a linear array of multiple biosensors. With this method, the PDE model is approximated by a system of ordinary differential equations.

  In Sec.\ref{sec:II}, the PDE model for a general reactive surface is described. It is followed by the derivation of the multi-compartment ODE model in Sec.\ref{s1}. The asymptotic properties of the least squares estimator for the concentration of target molecules is studies in Sec.\ref{s2} using the multi-compartment model of Sec.\ref{s1}. Sec.\ref{result} presents the results for a protein-based biosensor.

  \section{PDE model for the fluid flow}
\label{sec:II}
The aim is to estimate the concentration of target molecules in a fluid system where the dynamics are described by an advection-diffusion PDE model.

 Consider a flow chamber with a rectangular cross section where a flow of target molecules flows past multiple surface-based biosensors along the length of the chamber. There are $N$ identical biosensors which form a linear array along the flow direction on the surface of the chamber floor. We introduce three-dimensional Cartesian coordinates $(x,y,z)$ with the $y$-axis along the
flow direction and the $z$-axis along the height of the flow chamber and perpendicular to the surface of the biosensors. Biosensor $i$, for $i=1,2,\ldots ,N$, is located in the range $\left [ y_{i,1},y_{i,2} \right]$ along the $y$-axis and $\left[0,w\right]$ along the $x$-axis. The inlet of the flow chamber lies in the $x-z$ plane. The system is symmetric about the $x$-axis since the ratio of the height to the width of the flow chamber is selected to be less than $1/20$ \cite{hw}. The dimensions of the flow chamber and biosensors are:
\begin{align}\label{dimensions}
&\text{Flow chamber: }\text{Height}=h,\quad \text{Length}=l, \quad \text{Width}=w,\\ \nonumber
&\text{Biosensors: }\text{Length}=L, \quad \text{Spacing}=d,\\ \nonumber
&\text{Biosensor }i\text{ is located in the range } y \in \left [ y_{i,1}, y_{i,2}\right],\quad 0 \leq y_{i,1}, y_{i,2} \leq l.
\end{align}
  A flow chamber with $N=3$ biosensors is illustrated in Fig.\ref{fig0}. 
%
When target molecules in the solution arrive at the biosensors, chemical reactions are initiated which result in a change in impedance that is translated to change in the measured current. 

Below, an advection-diffusion PDE is used to describe the spatio-temporal evolution of concentration of target molecules in the flow chamber. It is coupled with a set of ODEs on the boundary which describes the adsorption of target molecules on the biosensors as a result of chemical reactions. 

{\bf Fluid flow dynamics:} The concentration of target molecules in the flow chamber (\ref{dimensions}), denoted by $A(t,y,z)$ is governed by an advection-diffusion PDE \cite{bib6} 
%
\begin{equation}\label{e1}
\frac{\partial A}{\partial t}=\gamma \left( \frac{{\partial} ^2 A}{\partial y^2} + \frac{{\partial} ^2 A}{\partial z^2} \right ) - v(z) \frac{\partial A}{\partial y}, \quad y \in (0,l), \quad z \in (0,h).
\end{equation}
Here $\gamma$ is the diffusion constant of the target molecule and $v(z)$ is the flow velocity in $y$ direction. The flow is assumed to be laminar and fully developed with a parabolic velocity profile defined by
\vspace{-4mm}
\begin{equation} \label{vel}
v(z)=4 \bar{v} (z/h)(1-z/h),
\end{equation}
where $\bar{v}$ is the maximum velocity \cite{Extending}.
\begin{figure}[t]
 \begin{center}
  \includegraphics[trim=0 1.5cm 0 1.5cm,width=.5\textwidth]{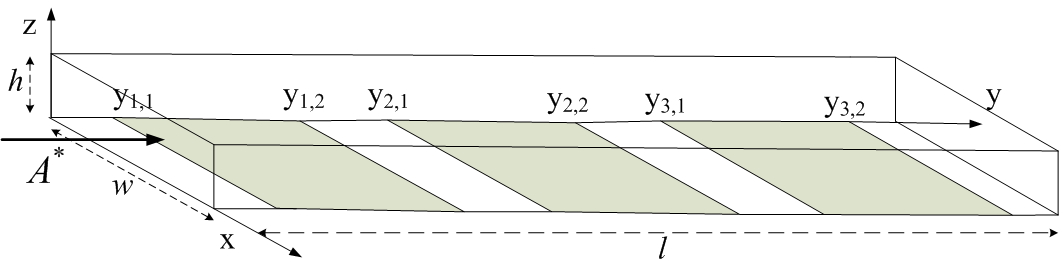}
   \caption{\label{fig0}\small{An equally spaced linear array of three biosensors in a rectangular flow chamber. The fluid containing target molecules enters from the left side. The concentration of target molecules at the inlet is $A^*$ as expressed in the boundary condition (\ref{inlet}).}}
   \end{center}
\end{figure}
Initially, the flow chamber is empty and the concentration inside the flow chamber is zero. So the initial condition is written as
\begin{align}\label{initial}
A(t=0,y,z)=0, \quad y \in (0,l), \quad z \in (0,h).
\end{align}

To complete the modeling of the fluid flow, it is necessary to specify the following Dirichlet and Neumann boundary conditions: The concentration at the inlet of the flow chamber is constant during the estimation process and equal to $A^*$. At the outlet of the flow chamber, there is no diffusive flux. There is insulation at the ceiling of the flow chamber at $z=h$. The gaps between the sensors on the floor of the flow chamber, at $z=0$, are also insulated. These boundary conditions are described as
\begin{equation} \label{inlet}
A(t,y=0,z)=A^{*},~\left . \frac{\partial A}{\partial y} \right |_{y=l}=0,~\left . \frac{\partial A}{\partial z} \right |_{z=h}=0,~\left .  \frac{\partial A}{\partial z} \right |_{z=0,\large{y \notin {\cup}_{i=1}^N [y_{i,1},y_{i,2}]}}=0.
\end{equation}
On the surface of each biosensor, the adsorption flux of target molecules is equal to the rate of consuming target molecules by the reactions. 

 {\bf Chemical dynamics}: Assuming that biosensor $i$, $i \in \{1,2,\ldots,N\}$, is located in  $y \in [y_{i,1},y_{i,2}]$ on y-axis, the corresponding boundary condition is expressed as \cite{bib6}
\begin{equation}\label{R}
\left . \gamma \frac{\partial A}{\partial z} \right |_{z=0,\large{y \in \left [y_{i,1}, y_{i,2} \right ]}}=R(A(t,y,z=0),\textbf{u}_i(t,y)),
\end{equation}
where the vector $\textbf{u}_i(t,y)$ contains the values of concentration of chemical species at time $t$ and location  $y$ on biosensor $i$. $R(A,\textbf{u}_i)$ is the rate of adsorption of target molecules per unit area on the biosensor surface. The rate of adsorption at each point is a function of the concentration of target molecules and chemical species on that point on the biosensor.	
The dynamics of the chemical species of biosensor $i$ at location $y$ are described by a system of ODEs as
\begin{equation}\label{u_n}
\hspace{-1mm}\frac{d\textbf{u}_i(t,y)}{dt}=G(\textbf{u}_i(t,y),{A(t,y,0)}),\quad  t > t_i,~ {\textbf{u}}_i \left ( t_i ,y\right )=u_0, \quad \text{for} \quad y \in \left[ y_{i,1}, y_{i,2}\right].
\end{equation}
Here, $G(\cdot)$ is a function which is described by the rate law of reactions on the biosensor. $t_i$ is the time instant at which the flow reaches biosensor $i$ and the biosensor starts responding. The rate of change of $\textbf{u}_i(t,y)$  depends on the concentration $A(t,y,z=0)$ of target molecules on the biosensor. The constant $u_0$ is the initial concentration of chemicals on each biosensor.  

{\bf Aim: } The aim is to estimate the concentration $A^{*}$ at the inlet of the flow chamber which appears in the boundary condition (\ref{inlet}). After describing the biosensor array model, statistical estimation algorithms are given in Sec.\ref{s1} and Sec.\ref{s2} to estimate $A^*$ given noisy measurements from the biosensors.

{\bf Measurement equation}: Finally, the measurement equation is specified. The response of biosensor $i$ at time $t$  is denoted by $g_i(A_1,t)$ where $A_1$ refers to the concentration at the inlet.  The response $g_i(A_1,t)$ can be described as a function of $\bar{\textbf{u}}_i(t)$ which denotes the surface average of the concentration vector $\textbf{u}_i(t,y)$ on biosensor $i$;
\begin{equation}\label{response}
g_i(A_1,t)=F\left (\bar{\textbf{u}}_i(t) \right ),\quad \bar{\textbf{u}}_i(t)=\frac{1}{y_{i,2}-y_{i,1}} \int_{y_{i,1}}^{y_{i,2}} \textbf{u}_i(y,t)\,dy.
\end{equation}
Here, $F(\cdot)$ is the transducer function which translates the concentration quantities on the biosensor to a corresponding electrical signal. In Sec.\ref{result}, we give a specific example of an actual biosensor where $F(\cdot)$ models the  conductance of the biosensor.  Considering the PDE model (\ref{e1})-(\ref{u_n}), the measurement taken at biosensor $i$ at time $t^{i,k}$, denoted by $m_i^k$, is
 \begin{equation}\label{measure}
{m}_i^k=g_i(A^*,t^{i,k})+{n}_i^k, \quad i \in \{ 1,2,\ldots ,N \},\quad k \in \{ 1,2,\ldots , S \}.
\end{equation}
Recall $A^*$ is the value of the concentration $A_1$ at the inlet and $n_i^k$ is the corresponding measurement noise. In (\ref{measure}), $S$ is the number of measurement samples taken at each biosensor. The noise samples $n_i^k$ for $i=1,\ldots,N$ and $k=1,\ldots,S$ are independent normally distributed with zero mean and finite variance $\sigma^2$. In (\ref{measure}), $g_i(A^*,t^{i,k})$ is an implicit function of the concentration $A^*$ through the PDE model (\ref{e1})-(\ref{u_n}).
%

\begin{figure}[t]
 \begin{center}
  \includegraphics[trim=.5cm 1cm 0cm 1.4cm, width=0.5\textwidth]{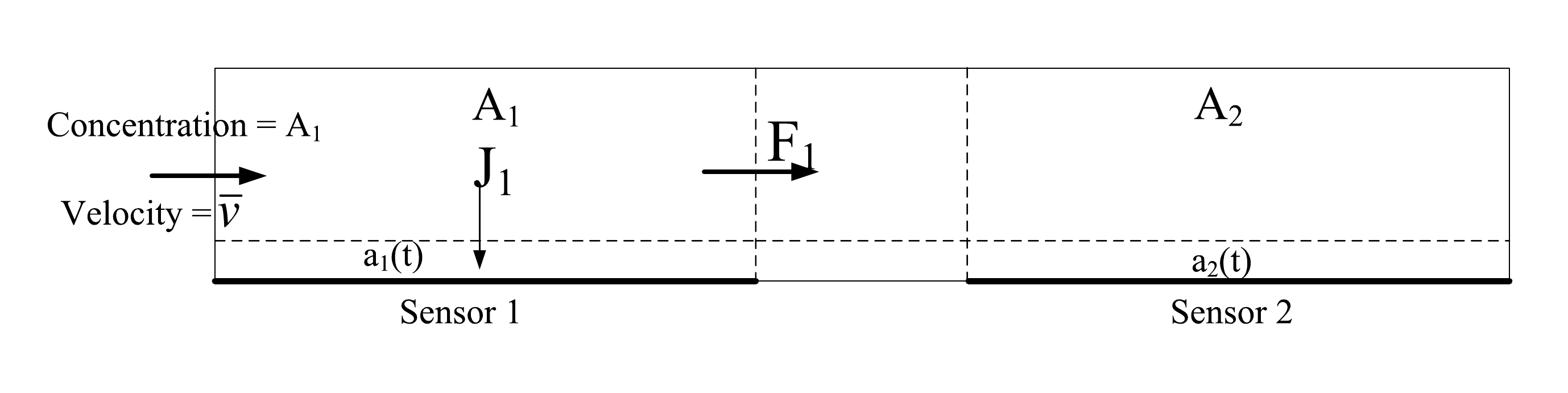} 
    \caption{\label{fig1} \small{Four-compartment model for two biosensors.}}
    \end{center}
\end{figure}
%
\section{Multi-compartment model approximation}
\label{s1}
Given the measurement equation of (\ref{measure}) and the PDE model of Sec.\ref{sec:II} defined by (\ref{e1})-(\ref{u_n}), the aim is to estimate the concentration $A^{*}$ at the boundary in (\ref{inlet}). The PDE is coupled with a set of ODEs through the boundary condition (\ref{R}) and cannot be solved analytically. To estimate $A^*$ in (\ref{inlet}), in this section, a multi-compartment ODE model is introduced that approximates the PDE by a system of ODEs.

The multi-compartment model is an extension of the existing two-compartment model to a new model for mass transport-binding experiments on a linear array of multiple biosensors. Its derivation is based on the multiple time-scale behaviour of the system \cite{singular}  and the divergence theorem. First, the  two-compartment model is reviewed in Sec.\ref{subsecA}. Then, the multi-compartment model is then derived in Sec.\ref{subsecB}.
\subsection{Review of the two-compartment model}
\label{subsecA}
The two-compartment model consists of a set of coupled ordinary differential equations which are used to analyze a variety of binding experiments influenced by mass transport \cite{Extending}. It is derived based on the multiple time-scale behaviour of the system and models the slow response of the system.
When the intrinsic reaction rates are comparable to or faster than the rate of transport of molecules to the reactive surface, a depleted region forms on top of the reactive surface where the amount of target molecules is slowly varying comparing to the concentration in the bulk. Because of this two time-scale behaviour, the flow chamber on top of the biosensor is divided vertically into two compartments, as shown in Fig.\ref{fig1}, in order to consider their dynamics separately. This model ignores the brief transitions before the bulk (outer compartment) concentration falls or rises to the concentration $A^{*}$ at the inlet. 
The concentration of target molecules within each compartment
is assumed to be spatially homogeneous. 
The concentration in the outer compartment, denoted by $A_1$, is equal to the concentration $A^{*}$  at the inlet of the flow chamber. The dynamics of the average concentration of target molecules in the inner compartment, denoted by $\bar{a}_1(t)$, is described by \cite{cdc}
\begin{align}\label{a_1}
h_{0}\frac{d\bar{a}_1(t)}{dt} = - R(\bar{a}_1(t),{\bar{\textbf{u}}_1(t)})+ \frac{\gamma}{h_0}\left ({A_1}-\bar{a}_1(t) \right),\quad, t >0, \quad \bar{a}_1(0)=0,
\end{align}
where
\vspace{-5mm}
\begin{equation}\label{hh0}
  h_0=\frac{1}{1.464} \left[\frac{\gamma h L}{\bar{v}}\right]^{1/3}
  \end{equation}
   is the height of the inner compartment. In (\ref{hh0}), $\gamma$ is the diffusion constant, $\bar{v}$ is the maximum velocity in the velocity profile (\ref{vel}), and $R(\cdot,\cdot)$ is defined in (\ref{R}). The average concentration of chemicals on the biosensor, denoted by $\bar{\textbf{u}}_1(t)$, is governed by
\vspace{-2mm}
\begin{equation} \label{u_1-two-comp}
\frac{d \bar{\textbf{u}}_1(t)}{dt} = G({\bar{\textbf{u}}}_1(t),\bar{a}_1(t)),\quad t>0,\quad \bar{\textbf{u}}_1(0)=u_0.
\end{equation}
%
%
\subsection{The multi-compartment model}
\label{subsecB}
In this section, the PDE model (\ref{e1})-(\ref{u_n}) is approximated by a system of ordinary differential equations. The flow chamber is partitioned into a series of two-compartment blocks above biosensors which are connected by middle compartments as shown in Fig.\ref{fig1}. As shown in Theorem~\ref{multiple-comp} below, the response of each biosensor can be described by an individual two-compartment model with a different concentration in its outer compartment. We apply the divergence theorem to the advection-diffusion PDE of (\ref{e1}) in the outer compartment associated with each biosensor  to find the concentration at the inlet of the next biosensor. By exploiting the different time-scale dynamics of the concentration in the flow chamber and applying some perturbation methods, a two-compartment model for the next biosensor can be derived. The concentration in the outer compartments of consecutive biosensors are related by (\ref{Ai}) in Theorem~\ref{multiple-comp}. For the derivation of the multi-compartment model, the following assumptions are required:
\begin{enumerate}[(1)]
\item The concentration $A(t,y,z)$ in the flow chamber is increasing and concave in $t$ for $y<t \bar{v}$ where $\bar{v}$ is the maximum velocity in the velocity profile (\ref{vel}).
\item The response of biosensor $i$ commences at time $t_i$ where 
\begin{align}\label{ti}
 t_i=\frac{y_{i,2}}{\bar{v}},~i=1,2,...,N.
\end{align}
Thus, $t_i$ is the time instant at which the flow reaches the far end of biosensor $i$ at $y=y_{i,2}$ in the flow chamber.
\item The spacing of biosensor $i$ is sufficiently small such that $\gamma^{1/3}(t_i-t_{i-1})=O(\gamma^{1/3})$ as $\gamma \to 0$ for $i \in \{ 1,2, \ldots ,N\}$. Here, $O(.)$ denotes the Landau big-O.
\end{enumerate} 
Assumption (1) models the physical reality that the concentration $A(t,y,z)$ at each point is bounded and eventually asymptotes at the concentration $A^*$. Therefore, $A(t,y,z)$ is concave in time after a certain time instant. 
Assumption (2) models the two-time scale behaviour of the flow chamber: the concentration in the outer compartment of each biosensor evolves rapidly compared to that in the inner compartment (vicinity of the surface of the biosensor). Assumption (3) reflects the fact that the spacing of the biosensors should be sufficiently small such that each biosensor is affected by the depletion region that is generated by the previous one. Otherwise, the biosensors have identical responses. The above assumptions are justified by numerous experimental studies of the biosensor, see \cite{bib6}.
 The following multi-compartment characterization is the main result of this section. 
 
\begin{theorem}
\label{multiple-comp}
Consider a flow of target molecules over an equally spaced linear array of $N$ identical biosensors in the flow chamber (\ref{dimensions}). Suppose the concentration of target molecules at the inlet of the flow chamber is a constant denoted by $A^*$. The concentration of target molecules $A(t,y,z)$ and chemical species are described by the PDE model (\ref{e1})-(\ref{u_n}). Under assumptions (1) and (2), as $\gamma \to 0$, there exists a time instant $t^*$ such that for $t \in (t_i,t^*)$, the dynamics of the average of the surface concentration of chemical species on biosensor $i$, denoted by $\bar{\textbf{u}}_i(t)$, satisfies 
\begin{align}\label{eq7}
h_{0}\frac{d\bar{a}_i(t)}{dt} &= \frac{\gamma}{h_0}\left ({A_i}-\bar{a}_i(t) \right)- R(\bar{a}_i(t),{\bar{\textbf{u}}}_i(t)) + O(\gamma^{4/3}),\quad t \in (t_i,t^*)\nonumber \\
\frac{d {\bar{\textbf{u}}}_i(t)}{dt} &= G({\bar{\textbf{u}}}_i(t),\bar{a}_i(t) + O(\gamma^2)),\quad t \in (t_i,t^*), \nonumber \\ 
\bar{a}_i(t_i)&=0, \quad \bar{\textbf{u}}_i \left ( t_i \right ) =u_0,\quad i = 1,\ldots,N.
 \end{align}
 Recall $\gamma$ in (\ref{eq7}) denotes the diffusion constant. The concentration in the flow chamber for $z \in (h_0,h-h_0)$ can be expressed as
 \begin{align}\label{jikjik}
A(t,y,z)=A_i + O(\gamma),\quad y \in (y_{i-1,2},y_{i,2}),~z \in (h_0,h-h_0),~t \in (t_i,t^*),
\end{align}
 where $h_0$ is defined in (\ref{hh0}).
   $A_i$ is a constant obtained by the following recursion:
   \begin{align}\label{Ai}
A_i=\alpha A_{i-1},~i=2,\ldots,N,~ \alpha=\left( 1- \frac{3 \gamma L}{2 h_0  \bar{v} h}\right),~A_1=A^*.
\end{align}
  \end{theorem}
    \begin{proof} The proof is in Appendix \ref{multi_comp}. \end{proof}
    
   The implication of the above theorem is that the multi-compartment model provides an accurate description of the dynamics of the concentration of target molecules. In (\ref{eq7}), Theorem~\ref{multiple-comp}, $h_0$ is the height of the inner compartment above each biosensor and $\bar{a}_i(t)$ denotes the spatial average of concentration in the inner compartment of biosensor $i$. The non-negative constant $A_i$  denotes the concentration in the outer compartment of biosensor $i$. The ODEs of the two-compartment model in (\ref{a_1}) and (\ref{u_1-two-comp}) are generalized for biosensor $i$ in (\ref{eq7}) in this model. 
  \section{Asymptotic analysis of the least squares estimator of the concentration}
\label{s2}
With the above multi-compartment characterization of the concentration of target molecules, this section deals with estimating the initial concentration $A^*$. The estimation is formulated as a least squares problem for the multi-compartment model. Then the asymptotic behaviour of the estimator in terms of consistency and asymptotic normality is investigated.
 Also an approximate formula for the variance of the finite-sample estimator  is obtained. This allows us to evaluate qualitatively how the variance of the estimate varies with the number of biosensors.
 

The concentration $A_1=A^*$ at the inlet of the flow chamber is estimated using  non-linear regression. The estimate denoted by $\hat{A_1}$ is defined as
\begin{equation}\label{eq70}
\vspace{-2mm}
\textstyle
\hat{A_1}= \underset{{A_1}\in {R^+}}{\text{arg min }} S^{-1} \displaystyle\sum_{i=1}^N \displaystyle\sum_{k=1}^{S}\left ({m}_i^k-{g}_i (A_1,t^{i,k}) \right)^2 
\end{equation}
where $N$ and $S$ refer to the number of biosensors and the number of time samples, respectively. In (\ref{eq70}), $m_i^k$ refers to $m_i(t^{i,k})$ which, according to  (\ref{measure}), is the measurement of biosensor $i$ taken at time $t^{i,k}$. ${g}_i (A_1,t^{i,k})$ defined in (\ref{response}), is the response of biosensor $i$ at time $t^{i,k}$ when the concentration at the inlet is $A_1$. Regarding (\ref{response}), ${g}_i (A_1,t^{i,k})$ is a function of $\bar{\textbf{u}}_i(t)$. Thus, the estimate is the solution of the optimization problem (\ref{eq70}) together with (\ref{eq7}). 

The following definition, assumptions and theorems establish strong consistency and asymptotic normality of the least squares estimator for the concentration $A^*$. 

\begin{definition}
\label{tail}
Let ${f}=({f}_t)$ and ${g}=({g}_t)$ be two sequences of real valued functions on $\Theta$ and $h$ be a function on $\Theta \times \Theta$. If $n^{-1}\sum_{t=1}^n f_t(\alpha) g_t(\beta)$, as $n \to \infty$, converges uniformly to $h(\alpha,\beta)$ for all $\alpha$ and $\beta$ in $\Theta$, $h=[f,g]$ is called the tail cross product of $f$ and $g$. 
\end{definition} 

Assume that:

(a) a sequence of real valued $N \times 1$ vectors $\textbf{y}_t$ has the structure $\textbf{y}_t=\textbf{f}_t(\theta_0)+\textbf{e}_t$ for $t=1,2,3,\ldots$ where the elements of the vector function $\textbf{f}_t$, denoted by $f_{i,t}$ for $i=1,2,\ldots ,N$, are known continuous functions on a compact subset $\Theta$ of a Euclidean space and the vectors $\textbf{e}_t$ for $t=1,2,3,\ldots$ are independent identically distributed with zero mean and finite covariance matrix $\sigma^2 I_N>0$. Here, $I_N$ denotes the $N \times N$ identity matrix.

(b) considering Definition~\ref{tail}, the tail cross product of $f_i=(f_{i,t})$, denoted by $[f_{i},f_{i}]$, for $i=1,2,\ldots ,N$ exists.  
 Besides, $Q(\theta)=\lim_{n \to \infty} n^{-1}  \sum_{t=1}^n \left | \textbf{f}_t(\theta_0)-\textbf{f}_t(\theta)\right |^2$ has a unique minimum at $\theta=\theta_0$. Here, $|.|^2$ denotes the L-2 vector norm.
 
 Any vector $\hat{\theta}_n$ in $\Theta$ which minimizes $Q_n(\theta)=n^{-1} \sum_{t=1}^n \left| \textbf{y}_t-\textbf{f}_t(\theta)\right|^2$, is a least squares estimate of $\theta_0$ based on  the first $n$ values of $\textbf{y}_t$. The following theorem establishes strong consistency of $\hat{\theta}_n$. 
 
 \begin{theorem}
\label{consis_in_paper}
Suppose that $ (\hat{\theta}_n )$ is a sequence of least squares estimators. Under assumptions (a) and (b), $\hat{\theta}_n$ is a strongly consistent estimator of $\theta_0$.
\end{theorem}

\begin{proof}
The proof for a scalar valued function $f$ is given in \cite{asymptotic}. Extending the proof to a scalar valued function is straightforward.
\end{proof}

To establish the asymptotic normality of a sequence of least squares estimators of a scalar parameter, we need the derivatives ${f}_{i,t}^{\prime}(\theta)=(\partial / \partial \theta){f}_{i,t}(\theta)$ and ${f}^{\prime\prime}_{i,t}(\theta)=(\partial^2 / \partial \theta^2){f}_{i,t}(\theta)$ for $i=1,2,\ldots ,N$ and $t=1,2,3,\ldots $. Suppose that:

(c)  the derivatives ${f}_{i,t}^{\prime}$ and ${f}^{\prime\prime}_{i,t}$ for $i=1,2,\ldots ,N$ exist and are continuous on $\Theta$ and that all tail cross products $[f_i,f_i^{\prime}]$, $[f_i,f_i^{\prime\prime}]$, and $[f_i^{\prime},f_i^{\prime\prime}]$ for $i=1,2,\ldots ,N$ exist. 

(d) for each $\theta$ in $\Theta$, $a(\theta)$ is defined as
$a(\theta)= \lim_{n \to \infty} n^{-1}\sum_{t=1}^{n} \sum_{i=1}^N \left ( {f}_{i,t}^{\prime}(\theta) \right )^2$. The true parameter $\theta_0$ is an interior point of $\Theta$ and $a(\theta_0)$ is not zero.

The following theorem provides conditions for the asymptotic normality of a sequence of least squares estimators.

\begin{theorem}
\label{asym}
Suppose that $(\hat{\theta}_n)$ is a sequence of least squares estimators of a scalar parameter $\theta_0$. Under assumptions (a) through (d), $\hat{\theta}_n-\theta_0$ is asymptotically normal, i.e. $\sqrt{n} \left ( \hat{\theta}_n-\theta_0 \right) \to N(0, \sigma ^2 a(\theta_0)^{-1})$.
\end{theorem}
\begin{proof}
The proof for a scalar function $f$ is given in \cite{asymptotic}. Extending the proof to vector valued functions is straightforward.
\end{proof}

Using Theorem~\ref{consis_in_paper} and Theorem~\ref{asym}, strong consistency and asymptotic normality of the estimator $\hat{A}_1$, defined in (\ref{eq70}), is established in Theorem~\ref{consist} below. 

\begin{theorem}
\label{consist}
Consider the observation model (\ref{measure}), dynamics (\ref{eq7}), the relation of the biosensor response with the concentration of species in (\ref{response}) and the recursion (\ref{Ai}). Assume that the noise samples in (\ref{measure}) are independent identically distributed with zero mean and finite variance $\sigma^2$. Then the estimate $\hat{A}_1$, defined in (\ref{eq70}), has the following asymptotic properties:

\noindent 1. $\hat{A}_1$ is strongly consistent as the time sample size $S$ in (\ref{eq70}) grows.

\noindent 2. The estimation error $\hat{A_1}-A^*$ is asymptotically normal as $S \to \infty$;
\begin{align}\label{gamma}
\hspace{-2mm}\sqrt{S} \left( \hat{A_1} - A^*\right)\to N(0, \frac{\sigma^2}{\Gamma}),~
\Gamma = \lim_{S \to \infty } \frac{1}{S} \displaystyle\sum_{k=1}^S \displaystyle\sum_{i=1}^N \alpha^{2i-2} \left[ \frac{\partial{{g} (\alpha ^ {i-1} A^*,t^{i,k}-t_i)}}{\partial{A}} \right]^2,
\end{align}
where $g(A,t)$, defined in (\ref{response}), is the response of each biosensor when the concentration in its outer compartment is $A$. $\partial{{g} (\alpha^{i-1}A^*,t^{i,k}-t_i)}/\partial{A}$ is the value of $\partial{{g} (A,t)}/\partial{A}$ at $A=\alpha^{i-1} A^*$ and $t=t^{i,k}-t_i$. 
\end{theorem}

\begin{proof} The proof is given in Appendix~\ref{consistproof}. \end{proof}

\begin{corollary}
\label{cor}
Consider the observation model (\ref{measure}), dynamics (\ref{eq7}), the relation of the biosensor response with the concentration of species in (\ref{response}) and the recursion (\ref{Ai}). Assume that the noise samples in (\ref{measure}) are independent identically distributed with zero mean and finite variance $\sigma^2$. Suppose that $\sigma_{S,N}^2$ denotes the variance of the finite sample estimator $\hat{A}_1$, defined in (\ref{eq70}) for finite $S$. Then,
\begin{align}\label{approx_variance}
\hspace{-2mm}\sigma_{S,N}^2 \approx {\sigma^2}/ {\sum_{i=1}^{N}d_i },~
 d_i=\alpha^{2i-2}\displaystyle\sum_{k=1}^{S} \left[ \frac{\partial{{g} (\alpha ^{i-1}A^*,t^{i,k}-t_i)}}{\partial{A}} \right]^2,
 \end{align}
where 
$\frac{\partial g}{\partial{A}} (\alpha^{i-1}A^*,t^{i,k}-t_i)$ is the value of $\partial{{g} (A,t)}/\partial{A}$ at $A=\alpha^{i-1} A^*$ and $t=t^{i,k}-t_i$. 
\end{corollary}

The proof of Corollary~\ref{cor} follows straightforwardly from Theorem~\ref{consist}. Corollary~\ref{cor} shows how the variance of the estimate of the concentration $A^*$ varies with the number of biosensors. In the experiments involving the protein-based biosensor described in Sec.\ref{result}, the approximation (\ref{approx_variance}) is used to explain how the estimation variance varies with the number of biosensors. 
\section{Case-study: Ion channel biosensor}
\label{result}
In this section, the multi-compartment model of Sec.\ref{subsecB} is evaluated for a protein-based biosensor, namely the ion channel switched (ICS) biosensor that was constructed and described in \cite{bib5}. ICS biosensor is a generic biosensor that can detect low molecular weight drugs, large proteins and micro-organisms \cite{bib5} with low concentrations as low as 10 fMolar~\cite{bib6}. This biosensor incorporates artificial ion channels in a lipid bilayer. The PDE model of Sec.\ref{sec:II} is specified for this biosensor by describing the  corresponding chemical reactions and measurement equation in Sec.\ref{describe_ICS}. 
In Sec.\ref{sim_res_sec}, the multi-compartment model (\ref{eq7}) is applied to an ICS biosensor array. Comparison between its response with the response of the PDE shows that the multi-compartment model describes the system accurately. 
Finally,  in Sec.\ref{estimation_res}, numerical examples are given for estimating the concentration $A^*$.

\subsection{Dynamics of the flow on (Ion Channel Switch) ICS biosensor}
\label{describe_ICS}
In this section,  the operation of the ICS biosensor is outlined. Then, the PDE model of Sec.\ref{sec:II} is constructed for a linear array of ICS biosensors.

{\bf Chemical Dynamics}: Recall Sec.\ref{s1} gave a generic description of chemical reactions in (\ref{u_n}). Here, the specific chemical dynamics on ICS biosensor are described. More details on the construction and operation of this biosensor can be found in \cite{bib6} and \cite{bib5}.  To specify chemical dynamics, we first outline briefly  the structure and operation of the ICS biosensor. The ICS biosensor is a surface based biosensor comprised of a lipid bilayer where ion channels are infused. The inner lipid layer is tethered to a gold substrate. The ion channels within this layer are tethered whereas the ones in the outer layer diffuse freely. The flow of ions through a channel only occurs when a mobile channel in the outer layer aligns to a fixed channel in the inner layer to form a conducting dimer. The arrival of target molecule cross-links
antibodies attached to the mobile outer layer channels, to
those attached to tethered lipids. This anchors
the channels distant, on average, from their inner layer
partners. The expected number of dimers is thus decreased. The conductance of the biosensor is proportional to the concentration of the dimers. Therefore the arrival of target molecules results in decreasing the biosensor admittance. There are eight chemical species on the ICS biosensor \cite{bib6}. Therefore, the vector of concentration of chemical species $\textbf{u}$ has eight elements. The primary species include binding site \emph{b} with concentration $B$, free moving ion channel \emph{c} with concentration $C$, tethered ion channel \emph{s} with concentration $S$, and dimer \emph{d} with concentration $D$. Initially, free moving ion channels \emph{c}, tethered channel \emph{s}, dimers \emph{d} are in equilibrium through a reversible chemical reaction. The arrival of target molecules initiates six other reactions and the equilibrium shifts towards decreasing the dimer concentration. The target molecule binds to the primary species to form complexes \emph{w}, \emph{x}, \emph{y}, and \emph{z} with concentrations $W$, $X$, $Y$, and $Z$ according to the following chemical reactions
\vspace{-1mm}
\begin{align}\label{reactions}
\begin{array}{cccc}
a+b\rightleftharpoons ^{f_1}_{r_1} w  \quad &  a+c\rightleftharpoons ^{f_2}_{r_2} x \quad & w+c\rightleftharpoons ^{f_3}_{r_3} y \quad &
x+b\rightleftharpoons ^{f_4}_{r_4}y  \\ 
 c+s \rightleftharpoons^{f_5}_{r_5}  d \quad  &              a+d\rightleftharpoons ^{f_6}_{r_6} z \quad & x+s \rightleftharpoons ^{f_7}_{r_7} z  \quad &
\end{array}
\end{align}
Here, $f_j$ and $r_j$ for $j = {1, 2,\ldots , 7}$, respectively
denote the forward and backward reaction rate constants. The corresponding rate equations for the reactions (\ref{reactions}) are
\begin{align} \label{rate eq}
\begin{array}{cccc}
R_1 = f_1AB - r_1W  & R_2 = f_2AC - r_2X  &
R_3 = f_3WC - r_3Y  & R_4 = f_4XB - r_4Y \\ 
R_5 = f_5CS - r_5D  & R_6 = f_6AD - r_6Z &
R_7 = f_7XS - r_7Z  &
\end{array}
\end{align}
Define the vector of concentration of chemical species as $\textbf{u} = { \left [B, C, D, S, W, X, Y, Z  \right ]}^T$  and  
\begin{equation}\label{r(u,A)}
f(\textbf{u},A) =
{\left [ R_1,R_2,R_3,R_4,R_5,R_6,R_7 \right ]}^{T},
\end{equation}
where $(\cdot)^T$ denotes transpose. The variations of the concentration of species on biosensor $i$ can be expressed as \cite{cdc}
\vspace{-2mm}
\begin{equation}\label{new odes}
\frac{d\textbf{u}_i}{dt}=Mf(\textbf{u}_i,A) \quad \text{for} \quad t > t_i,\quad \textbf{u}_i(t_i)=\textbf{u}_0,\quad i=1,2,\ldots ,N.
\end{equation}
\vspace{-1mm}
where  $M$ is a $7 \times 7$ constant matrix obtained as \cite{cdc}. 
\begin{equation}\label{M}
M=\left [
\begin{array}{ccccccc}
-1&0&0&-1&0&0&0\\
0&-1&-1&-1&0&0&0\\
0&0&0&0&1&-1&0\\
0&0&0&0&-1&0&-1\\
1&0&-1&0&0&0&0\\
0&1&0&-1&0&0&-1\\
0&0&1&1&0&0&0\\
0&0&0&0&0&1&1\\
\end{array}
\right ]
\end{equation}
  The index $i$ in $\textbf{u}_i=\left [B_i~C_i~D_i~S_i~W_i~X_i~Y_i~Z_i \right ]$ refers to biosensor $i$. From the reactions (\ref{reactions}), the rate of adsorption of target molecules on the surface of biosensor $i$ can be written as $R(A,\textbf{u}_i)=f_1AB_i +f_2AC_i+f_6AD_i - r_1W_i- r_2X_i - r_6Z_i$. Thus the corresponding boundary condition in (\ref{R}) can be expressed as
\begin{eqnarray}\label{boundary_on_ICS}
\left . \gamma \frac{\partial A}{\partial z}\right |_{z=0}=A q^T \textbf{u}_i - p^T \textbf{u}_i\ , \quad i=1,2,\ldots ,N,
 \end{eqnarray}
where $q$ and $p$ are vectors which can be defined as $q=[f_1 ~ f_2 ~ f_6 ~ 0 ~ 0 ~0~ 0 ~ 0]^T$ and $p=[0 ~0 ~ 0 ~ 0 ~ r_1 ~r_2~ 0 ~ r_6]^T$.
Applying a small alternative potential between the gold
substrate and a reference electrode in the test solution generates a charge at the gold surface which causes
electrons flow through ion channels \cite{bib5}. The measured current in the external circuit is proportional to the surface average of dimer concentration. Denoting the average dimer concentration of biosensor $i$ by $\bar{D}_i(t)$, the observation equation, on biosensor $i$, can be written as $m_i(t)=\bar{D}_i(t)+n_i(t)$.
 %
%
%
Physical reality demands that the derived PDE model defined by (\ref{e1})-(\ref{inlet}), (\ref{new odes}), and (\ref{boundary_on_ICS}) has a non-negative solution since the solution corresponds to non-negative physical quantities. In the following, Theorem~\ref{pos} states that the solution is actually non-negative.
By proving the positivity of the solution, it can be verified that the PDE model (\ref{e1})-(\ref{inlet}), and (\ref{new odes}), (\ref{boundary_on_ICS}) is well defined.
\begin{theorem}
\label{pos}
Let (\ref{e1})-(\ref{inlet}), (\ref{new odes}), and (\ref{boundary_on_ICS}) describe the dynamics of $A(t,y,z)$ and $\textbf{u}_i(t)$ for $i=1,2,\ldots ,N$. Then, $\textbf{u}_i(t)\geq 0$ for $t \geq t_i$ and $A(t,y,z) \geq 0$ for $y \in [0,l]$, $z \in [0,h]$, and $t \geq 0$.
\end{theorem}

\begin{proof}
The proof can be found in Appendix~\ref{positivity}.
\end{proof}
\subsection{Illustration of the accuracy of the multi-compartment model}
\label{sim_res_sec}
This section considers a biosensor array comprising four ICS biosensors.  The aim is to show that the multi-compartment model of Sec.\ref{subsecB} yields an excellent approximation to the flow dynamics. The height and width of the flow chamber are $h=0.1$~mm and $w=2$~mm. The length of each biosensor is $L=2$~mm and the diffusion constant is equal to $\gamma=10^{-6}~\mathrm{cm}^2\mathrm{/s}$. 
We study the effect of varying the concentration $A^*$ at the inlet of the flow chamber in the range $10^{-11} \sim  10^{-8}\text{ Mol}/\text{m}^3$ and the velocity in the range $10 \sim 100 \mu \text{L}/\text{min}$. The notation $\text{Mol}/\text{m}^3$, throughout the paper, stands for mole per meter cube.  The spacing between biosensors is $d=1$~mm.

The response of biosensors obtained by the multi-compartment model (\ref{eq7}) is simulated and compared with the response of the PDE model (\ref{e1})-(\ref{inlet}), (\ref{new odes}), (\ref{boundary_on_ICS}). The Comsol multi-physics simulation software is used to solve the PDE via the finite element method. The predefined {\tt convection} and {\tt diffusion} application mode in Comsol is used to define the governing PDE in the domain. The ODEs (\ref{new odes}) on the boundary are defined through the {\tt weak form} boundary setting.
 
Since the measured output of the biosensor is a linear function of the average dimer concentration, define the normalized error between the ODE and PDE responses as
\vspace{-0.5mm}
\begin{equation}\label{norm_error}
e_i(t)=\left | \bar{D}_i(t)-\bar{D}_i^{\text{ODE}}(t) \right | /\bar{D}_i(t),\quad i=1,\ldots,N.
\end{equation}
\vspace{-0.5mm}
In (\ref{norm_error}), $\bar{D}_i(t)$ is the average dimer concentration on biosensor $i$, obtained by the PDE model (\ref{e1})-(\ref{inlet}), (\ref{new odes}), and (\ref{boundary_on_ICS}) and $\bar{D}_i^{\text{ODE}}(t)$ is the corresponding response from  the multi-compartment model (\ref{eq7}).
%
Fig.\ref{res1} shows the normalized error (\ref{norm_error}) versus time for two values of  concentration $A^{*}=10^{-11}$ and $A^{*}=10^{-8} \text{ Mol}/\text{m}^3$. It can be seen that the error during 1000 seconds of simulation time is less than $0.015 \%$ for $A^{*}=10^{-11} \text{ Mol}/ \text{m}^3$ and less than $8 \%$ for $A^{*}=10^{-8} \text{ Mol}/ \text{m}^3$ for all the biosensors.

Fig.\ref{errorvsf} shows the normalized error (\ref{norm_error}) for the first and second biosensor for different flow rates. The concentration of target molecules is $A^{*}=10^{-11}\text{ Mol}/ \text{m}^3$. The figure shows that by increasing the flow rate to $100 \mu \text{L}/\text{min}$, the multi-compartment model (\ref{eq7}) remains accurate within $9\%$ error for the first and second biosensor.
\begin{figure}[t]
 \begin{center}
  \includegraphics[trim=0cm 1cm 0cm 0cm,width=0.9\textwidth]{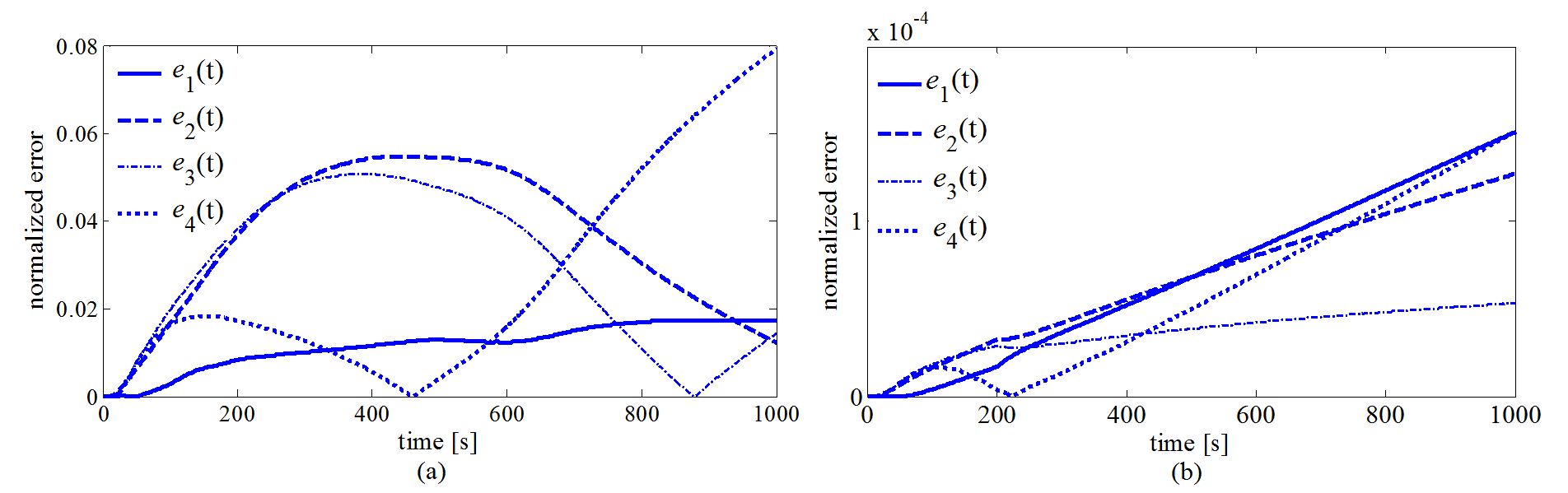}
  \caption{\label{res1}\small{The multi-compartment ODE model (\ref{eq7}) is compared with the PDE model  (\ref{e1})-(\ref{inlet}), (\ref{new odes}), and (\ref{boundary_on_ICS}) by plotting the normalized error described in (\ref{norm_error}) for four biosensors. Plot (a) shows the results for $A^{*}=10^{-8}\text{ Mol}/ \text{m}^3$ and plot (b) corresponds to $A^{*}=10^{-11}\text{ Mol}/ \text{m}^3$. The flow rate is set to $10~\mu\mathrm{L/min}$. The length of the biosensors is $L=2$~mm and their spacing is $d=1$~mm.}}
\end{center}
%
\end{figure}

\begin{figure}[t]

 \begin{center}
  \includegraphics[trim=0cm 2cm 0cm 0cm, width=0.9\textwidth]{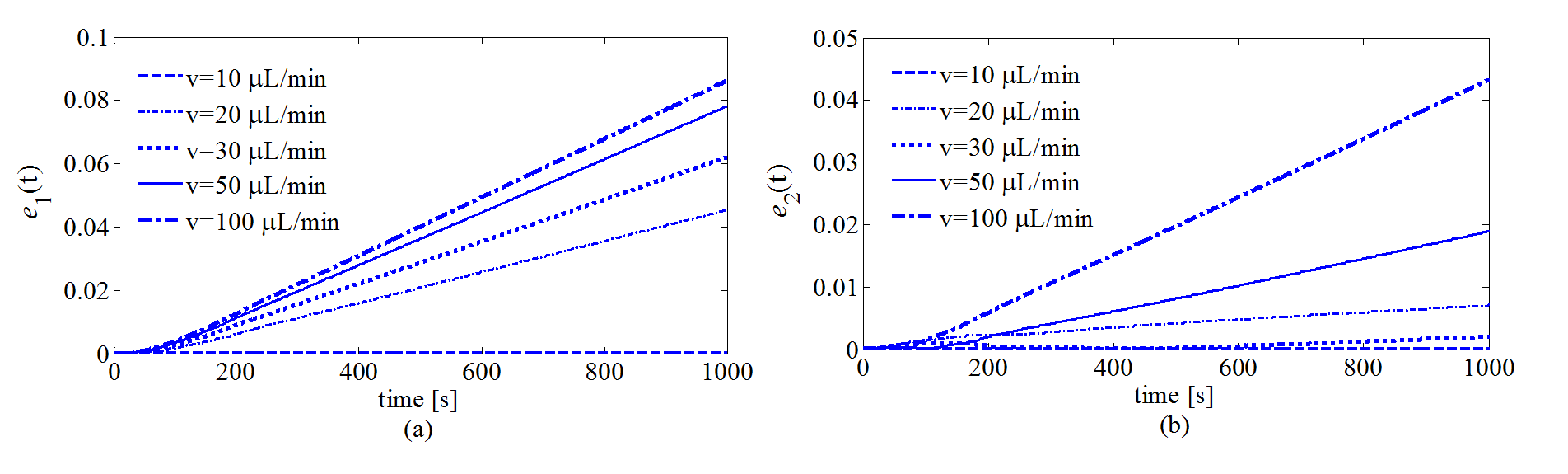}
  \caption{\label{errorvsf}\small{The multi-compartment ODE model (\ref{eq7}) is compared with the PDE model (\ref{e1})-(\ref{inlet}), (\ref{new odes}), (\ref{boundary_on_ICS}) for different flow rates. The normalized error (\ref{norm_error}) is plotted for the first (plot (a)) and second biosensor (plot (b)). All the results correspond to $A^{*}=10^{-11} \text{ Mol}/ \text{m}^3$. The length of biosensors is $L=2$~mm and the spacing is $d=1$~mm.}}
  \vspace{-2mm}
  \end{center}
  \end{figure}

\subsection{Investigating the properties of the estimator}
\label{estimation_res}
This section compares the error variance of the estimate $\hat{A}_1$ (using the approximation (\ref{approx_variance})) with numerical simulations.
Here, the analytical results of Sec.\ref{s2} for the ICS biosensor array is compared with the corresponding simulated results. The achievable improvement in the estimate based on the number of biosensors is also evaluated.

The variance $\sigma_{S,N}^2$ of the finite-sample estimator $\hat{A_1}$ (\ref{eq70}) is estimated by Monte Carlo simulations. The results are then compared with the approximate value (\ref{approx_variance}) in Corollary~\ref{cor} for verification. The standard deviation of $\hat{A_1}$ for different number of biosensors is shown in Table \ref{monte}. The variance is obtained when $S=300$ samples with sampling rate $1 \text{ sample}/\text{s}$  are used for estimation. The actual value of the concentration is $A^*=10^{-8}\text{ Mol}/\text{m}^3$. The signal to noise ratio, defined as the ratio of initial dimer concentration squared to the noise variance, is set to $10$ dB. Table \ref{monte} shows that the estimation variance, for $N=1,2,3$ biosensors, is decreasing slightly less than $1/N$ as the number of biosensors increases. The reason is that $d_i$, defined in (\ref{approx_variance}), is increasing for $i \in \{ 1,2,3\}$ in this case. In order to justify this result, the behaviour of the response $g(A,t)$, defined in (\ref{response}), is investigated. Fig.\ref{sense_fig} illustrates the response of ICS biosensor $g(A,t)$ in (\ref{response}) versus the concentration $A$ of target molecules  in the outer compartment. Recall that the measured output
current of the biosensor is proportional to the average dimer concentration which decreases as the concentration $A$ of target molecules increases. Fig.\ref{sense_fig} shows that the response curve is  sigmoidal and saturates for low and high concentration. As a result, there exists a concentration $A_c(t)$ for each time instant $t$ such that in the range $A > A_c(t)$, the response $g(A,t)$ is convex in $A$. Therefore, the derivative $\partial g (A,t)/\partial A$ is increasing. Since $\partial g (A,t)/\partial A$ is negative, $\left[ \partial g (A,t)/\partial A \right ]^2$ is decreasing in $A$ for $A>A_c(t)$. Define $H(A)$ as $H(A)=\sum_{k=1}^S\left[ \partial g (A,t^k)/\partial A \right ]^2$. Then, there is a concentration $A_c^{*}$ such that in the range $A>A_c^*$, $H(A)$ is decreasing. Assume that the sampling time points $t^{i,k}$ for each biosensor are selected such that the time difference $t^{i,k}-t_i$ in (\ref{approx_variance})  is constant and equal to $t^k$ for $i \in \{1,2, \ldots , N\}$. Then, the sensitivity $d_i$, defined in (\ref{approx_variance}), can be written as $d_i=\alpha^{2i-2} H(\alpha^{i-1}A^*)$. The behaviour of $H$ can explain the ascending behaviour of $d_i$ in Table \ref{monte}. To this end, $H(A)$ and $\alpha^2 H(\alpha A)$ are illustrated versus $A$ in Fig.\ref{ajaba}. From the experimental values of the parameters, $\alpha$ in (\ref{Ai}) is obtained as $\alpha=0.8173$. The number of time samples in acquiring $H(A)$ is $S=300$ and the sampling rate is $1\text{ sample}/\text{s}$. Fig.\ref{ajaba} shows that there is a range $[A^{\text{m}},A^{\text{n}}]$ such that
\begin{align}\label{H}
\alpha^2 H(\alpha A)> H(A), \quad A \in [A^{\text{m}},A^{\text{n}}], \quad H(A)=\sum_{k=1}^S\left[ \partial g (A,t^k)/\partial A \right ]^2 .
\end{align}
 Consequently, $d_i$ in (\ref{approx_variance}) is increasing as long as $\alpha^{i-1}A^*$ belongs to $[A^{\text{m}},A^{\text{n}}]$ for $i \in {1,2,...,N-1}$. Assuming that $A^* \in [A^{\text{m}},A^{\text{n}}]$, there is a positive integer $N^*$ such that for $N \leq N^*$, $d_i$ in (\ref{approx_variance})    is increasing for $i=1,2, \ldots , N$. Hence, using $N \leq N^*$ biosensors decreases the estimation variance to less than $1/N$ of the variance obtained with a single biosensor. The expression for $N^*$ can be obtained as
\vspace{-4mm} 
 \begin{align}\label{Nstar}
 N^*=\lfloor{ \frac{\log{\frac{A^{\text{m}}}{A^*}}}{\log{\alpha}} }\rfloor +2, 
 \end{align}
 where $A^{\text{m}}$ is defined in (\ref{H}), $A^*$ is the concentration of target molecules at the inlet of the flow chamber, and $\alpha$ is defined in (\ref{Ai}). In (\ref{Nstar}), $\lfloor \cdot \rfloor$ denotes the floor function which maps a real number to the largest integer which is not larger than that real number. According to the value of $A^{\text{m}}$ in Fig.\ref{ajaba}, the value of $N^*$ is obtained as $N^*=4$ for $A^*=10^{-8}\text{Mol}/\text{m}^3$.
%
%
%

It can be seen from Fig.\ref{sense_fig} and Fig.\ref{ajaba} that increasing the number of biosensors beyond a certain number does not decrease the estimation variance significantly  since the derivative $\partial g (A,t)/ \partial A$ approaches zero as target molecules in the fluid are grabbed by the previous biosensors in the array. 
\begin{table}
\caption{\label{monte}\small{Comparison between the simulated and approximate value (\ref{approx_variance}), in Corollary~\ref{cor}, for the variance $\sigma_{S,N}^2$ of $\hat{A_1}$ (\ref{eq70}): The simulated and analytical values of the standard deviation $\sigma_{S,N}/A^*$ for $S=300$ time samples are shown. The sampling rate is $1 \text{ sample}/\text{s}$.
$A^*=10^{-8}\text{ Mol}/\text{m}^3$. The signal to noise ratio, defined as the ratio of initial dimer concentration squared to the noise variance, is equal to 10 dB.}}
\begin{center}
\begin{tabular} {c|c|c|}
\cline{2-3}
&\multicolumn{2}{|c|}{$\sigma_{S,N}/A^*$}\\
\cline{2-3}
&  Simulated & Analysis (\ref{approx_variance}) \\
\hline
\multicolumn{1}{|c|}{N=1}& 0.0971 & 0.0943 \\
\hline
\multicolumn{1}{|c|}{N=2} & 0.0611 & 0.0651 \\
\hline
\multicolumn{1}{|c|}{N=3} & 0.044 & 0.051\\
\hline
\end{tabular}

\end{center}
\vspace{-2mm}
\end{table}

\begin{figure}[t]
\begin{center}
\includegraphics[width=0.6 \textwidth]{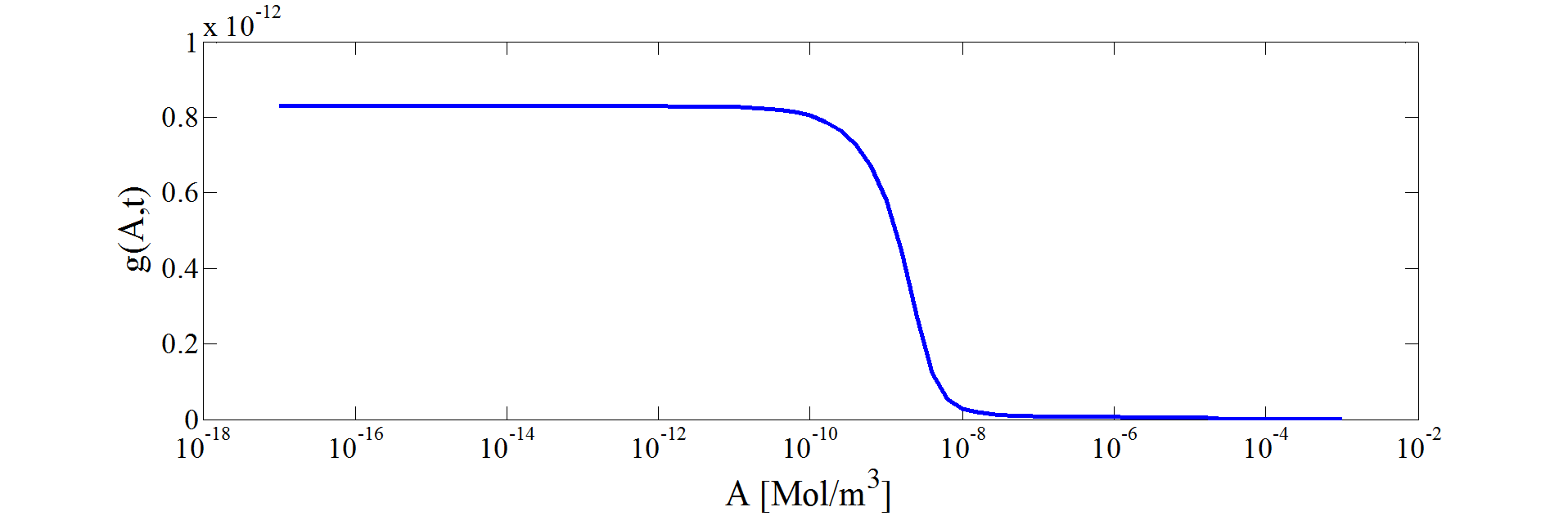}
\caption{\label{sense_fig}\small{The response $g(A,t)$, defined in (\ref{response}), is insensitive to low concentrations of target molecules. By increasing the concentration of target molecules, the sensitivity of $g(A,t)$ (the magnitude of the derivative $\partial g(A,t) / \partial A$) is initially increasing and then decreasing. The response is plotted for a time horizon of $500$s.}}
\vspace{-4mm}
\end{center}
\end{figure} 

\begin{figure}[t]
\begin{center}
\includegraphics[width=1 \textwidth]{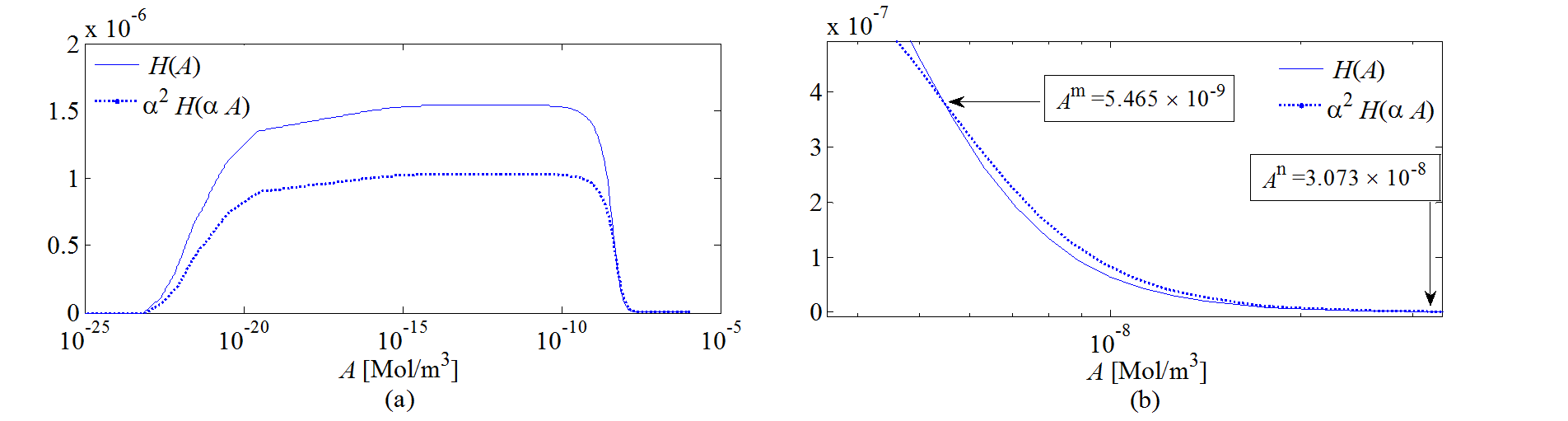}
\caption{\label{ajaba} \small{The ascending behaviour of $d_i$ in (\ref{approx_variance}) is investigated by comparing $H(A)$ and $\alpha^2H(\alpha A)$ where $H(A)$ is defined in (\ref{H}). The part of plot (a) where $\alpha^2H(\alpha A)> H(A)$ is magnified in plot (b). The interval $[A^{\text{m}},A^{\text{n}}]$ in (\ref{H}) is specified for this example. The value of $N^*$ in (\ref{Nstar}) is obtained as $N^*=4$ for the initial concentration $A^*=10^{-8}\text{Mol}/\text{m}^3$.}}
\end{center}
\end{figure}
The asymptotic behaviour of the estimate $\hat{A_1}$ (\ref{eq70}), using an array of $N$ biosensors, for $N=1,2,3$,  is illustrated in Table~\ref{tablejoon}. The asymptotic variance of $\sqrt{S} \hat{A_1}$, obtained as $\sigma^2 \Gamma^{-1}$ in (\ref{gamma}), is computed for $N=1,2,3$. Then the simulated value of $S \sigma^2_{S,N}$, where $\sigma^2_{S,N}$ is defined in Corollary~\ref{cor}, is compared with the corresponding asymptotic variance. The variance $\sigma^2_{S,N}$ is obtained by Monte Carlo simulations. The initial concentration is $A^*=10^{-8}\text{ Mol}/\text{m}^3$.
 In summary, it can be seen that:
\begin{remunerate}
\item The multi-compartment model (\ref{eq7})  predicts the response of the ICS biosensor array accurately over the range of concentrations between $10^{-11}\text{ Mol}/\text{m}^3$ and $10^{-8}\text{ Mol}/\text{m}^3$.
\item The initial concentration $A^*$ at the inlet of a flow chamber can be estimated with the ICS biosensor array using the multi-compartment model. Table~\ref{monte} shows that by using an array of three sensors and $300$ measurement samples during the first $300$s of the response, less than $5\%$ estimation error can be achieved.
\item When the initial concentration $A^*$ is in a certain range and the number of biosensors $N$ does not exceed a certain limit, the estimation variance is less than $1/N$ of the variance obtained by a single biosensor. This shows that there can be significant advantages in using an array of biosensors compared to a single biosensor.
\end{remunerate}
\begin{table}[t]
\caption{\small{Asymptotic behaviour of $\hat{A_1}$ defined in (\ref{eq70}): The finite-sample estimation variance $\sigma^2_{N,S}$ of $\hat{A_1}$ is obtained by Monte Carlo simulations. The concentration at the inlet of the flow chamber is $A^*=10^{-8}/\text{m}^3$. From (\ref{gamma}), $\sigma^2 \Gamma^{-1}$ is computed.}}
\begin{center}
\begin{tabular}{c|c|c|c|c|}
\cline{2-5}
 & \multicolumn{3}{|c|}{$S\sigma^2_{S,N}$} & \multicolumn{1}{|c|}  {\multirow{2}{*}{$\sigma^2 \Gamma^{-1}$}}  \\
 \cline{2-4}
  & $S=300$ & $S=1000$ & $S=5000$\\
  \cline{1-5}
  \multicolumn{1}{|c|}{$N=1$} & $2.8273\times 10^{-16}$&$7.4021 \times 10^{-16}$ & $2.3283 \times 10^{-15}$&$2.491 \times 10^{-12}$\\
  \cline{1-5}
  \multicolumn{1}{|c|}{$N=2$} & $1.1199 \times 10^{-16}$ & $3.2628 \times 10^{-16}$ & $1.5153 \times 10^{-15}$& $1.139 \times 10^{-12}$\\
  \cline{1-5}
  \multicolumn{1}{|c|}{$N=3$} & $5.808 \times 10^{-17}$ & $1.5981 \times 10^{-16}$ & $7.2200 \times 10^{-16}$ & $0.692 \times 10^{-12}$ \\
  \cline{1-5}
\end{tabular}
\end{center}

\label{tablejoon}

\end{table}

\appendix
%
\section{Proof of Theorem~\ref{multiple-comp}}
\label{multi_comp}
Divide the flow chamber into two segments along the $z$ direction at $z=h_0$ where $h_0$ is selected as (\ref{hh0}). Then discretize the $y$-axis at the edges of the biosensors located at $y=y_{i,1}$ and $y=y_{i,2}$, as shown in Fig.~\ref{fig1}. The inner compartment of biosensor $i$ is located at $y \in(y_{i,1},y_{i,2})$ and $z \in (0,h_0)$. The outer compartments above biosensor $i-1$ and $i$ are separated by a middle compartment which is located in the range $y \in (y_{i-1,2},y_{i,1}) $ and $z \in (h_0,h)$.

Non-dimensionalizing (\ref{e1}) yields
\begin{align}\label{nondim}
     \epsilon^2 &\frac{\partial a}{\partial \tau} = \epsilon^2 \frac{{\partial} ^2 a}{\partial Z^2} +  \frac{{\partial} ^2 a}{\partial Y^2} -  \frac{\epsilon Z(1-Z)}{p} \frac{\partial a}{\partial Y}, \quad
      \epsilon ={L}/{h}, \quad p=\frac{\gamma}{4h\bar{v}} \\ \nonumber
     a &=A/A^{*}, \quad \tau={\gamma}/{h^2}t  , \quad Z=z/h, \quad Y=y/L.
     \end{align}
     The P\'{e}clet number \cite{numerical} in this problem is 
\vspace{-2mm}
\begin{align}\label{pec}
P_e(Z)={\epsilon Z(1-Z)}{p}^{-1},
\end{align}
which is used as a criterion to identify the dominance of advection or diffusion in the transport of particles in the flow. 
Obtain the PDE (\ref{nondim}) at $Z=\epsilon_0/h$ as
\begin{align}\label{nondim2}
     \epsilon^2 & \left . \frac{\partial a}{\partial \tau}\right |_{Z=\epsilon_0/h} = \epsilon^2  \left . \frac{{\partial} ^2 a}{\partial Z^2}\right |_{Z=\epsilon_0/h} +  \left . \frac{{\partial} ^2 a}{\partial Y^2} \right |_{Z=\epsilon_0/h} -  \frac{\epsilon \frac{\epsilon_0}{h}\left(1-\frac{\epsilon_0}{h} \right)}{p} \left . \frac{\partial a}{\partial Y} \right |_{Z=\epsilon_0/h}, 
     \end{align}
Assume that $\epsilon_0=O(\gamma^2)$. Therefore, the P\'{e}clet number (\ref{pec})  at $Z=\epsilon_0/h$ is  $P_e(\frac{\epsilon_0}{h})=O(\gamma)$ and the advection term on the right hand side of (\ref{nondim2}) is negligible comparing to the diffusion term. Thus, (\ref{nondim2}) at $Z=\epsilon_0/h$ is rewritten as
\begin{align}\label{Astar}
     \epsilon^2 & \left . \frac{\partial a}{\partial \tau}\right |_{Z=\epsilon_0/h} = \epsilon^2  \left . \frac{{\partial} ^2 a}{\partial Z^2}\right |_{Z=\epsilon_0/h} +  \left . \frac{{\partial} ^2 a}{\partial Y^2} \right |_{Z=\epsilon_0/h} +  O(\gamma), 
     \end{align}
 Hence, the dimensional PDE (\ref{e1}) at $z=\epsilon_0$ can be expressed as
\begin{align}\label{22}
\vspace{-2mm}\left . \frac{\partial A}{\partial t} \right |_{z=\epsilon_0}=\gamma  \left . \frac{{\partial} ^2 A}{\partial z^2}\right |_{z=\epsilon_0} + \gamma \left . \frac{{\partial} ^2 A}{\partial y^2} \right |_{z=\epsilon_0} +  O(\gamma^2),
  \end{align}
  
  From the selection of $h_0$ in (\ref{hh0}), the magnitude of $h_0$ is of order $O(\gamma^{1/3})$. Assume that the value of $\epsilon_0$ is selected such that $0<\epsilon_0<h_0$. The second order derivative in $z$ direction in (\ref{22}) can be expressed as
  \begin{align}\label{sub1}
  \left . \frac{{\partial} ^2 A}{\partial z^2}\right |_{z=\epsilon_0}=\frac{1}{h_0}\left[ \left . \frac{{\partial} A}{\partial z}\right |_{z=h_0} - \left . \frac{{\partial} A}{\partial z}\right |_{z=0} \right]+ O(h_0)
  \end{align}
  The derivative $\left . \frac{\partial A}{\partial z} \right |_{z=h_0}$  in (\ref{sub1}) can be expressed as
\begin{align} \label{t2}
\left . \frac{\partial A}{\partial z} \right |_{z=h_0} = \frac{1}{{h_0}}\left[ A(t,y,h_0+\epsilon_0)-A(t,y,\epsilon_0)\right] + O(h_0).
\end{align}
Using the boundary condition (\ref{R}), the derivative $\left . \frac{\partial A}{\partial z} \right |_{z=0}$  in (\ref{sub1}) on sensor $i$ is
\begin{equation}\label{t1}
\left .\frac{\partial A}{\partial z}\right |_{z=0} = \frac{1}{\gamma} R \left( A(t,y,\epsilon_0)+O(\epsilon_0),\textbf{u}_i \right),\quad y \in(y_{i,1},y_{i,2}).
\end{equation}
Substituting (\ref{t2}) and (\ref{t1}) in (\ref{sub1}) yields
\begin{align}\label{sub2}
  \left . \frac{{\partial} ^2 A}{\partial z^2}\right |_{z=\epsilon_0}&=\frac{1}{h_0}\left[ \frac{1}{{h_0}}\left[ A(t,y,h_0+\epsilon_0)-A(t,y,\epsilon_0)\right] + O(h_0)\right] \\ \nonumber
  &- \frac{1}{h_0\gamma} R \left( A(t,y,\epsilon_0)+O(\epsilon_0),\textbf{u}_i \right) + O(h_0).~y\in(y_{i,1},y_{i,2})
  \end{align}  
 From (\ref{sub2}), the PDE (\ref{22}) for $y\in(y_{i,1},y_{i,2})$ can be expressed as
 \begin{align}\label{eehaa}
 h_0 \left . \frac{\partial A}{\partial t} \right |_{z=\epsilon_0}&=\frac{\gamma}{h_0}  \left[ A(t,y,h_0+\epsilon_0)-A(t,y,\epsilon_0)\right]-R \left( A(t,y,\epsilon_0),\textbf{u}_i \right)+O(\epsilon_0)\\ \nonumber
 &+\gamma O(h_0)+h_0 \gamma \left . \frac{{\partial} ^2 A}{\partial y^2} \right |_{z=\epsilon_0} + \gamma O(h_0^2)+h_0 O(\gamma^2), \quad y \in (y_{i,1},y_{i,2})
 \end{align}
 The exponent of the concentration $A$ in the polynomial $R(A,\textbf{u})$ can be equal or greater than one. Thus, the maximum order of approximation error in $R(A,\textbf{u})$ for an approximate $A$ occurs when the exponent of $A$ is one. This case is considered in (\ref{eehaa}) and $R \left( A(t,y,\epsilon_0)+O(\epsilon_0),\textbf{u}_i \right)$ is replaced with $R \left( A(t,y,\epsilon_0),\textbf{u}_i \right)+O(\epsilon_0)$. The value of $\epsilon_0$ is selected as $\epsilon_0=O(\gamma^{2})$.
From (\ref{hh0}),  $h_0=O(\gamma^{1/3})$. Thus, we have 
$\frac{\gamma}{h_0}=O(\gamma^{2/3})$, $\gamma h_0=O(\gamma^{4/3})$, $\gamma h_0^2=O(\gamma^{5/3})$, and $h_0\gamma^2=O(\gamma^{7/3})$.
Thus, the sum of the last five terms on the right hand side of (\ref{eehaa}) is of order $O(\gamma^{4/3})$. The first term on the right hand side of (\ref{eehaa}) is of order $O(\gamma^{2/3})$ and the left hand side of (\ref{eehaa}) is of order $O(\gamma^{1/3})$. Equation (\ref{eehaa}) for $y \in (y_{i,1},y_{i,2})$ can be rewritten as
\begin{align}\label{eehaa2}
\vspace{-1mm}
 \hspace{-2mm} h_0 \left . \frac{\partial A}{\partial t} \right |_{z=\epsilon_0}=\frac{\gamma}{h_0}  \left[ A(t,y,h_0+\epsilon_0)-A(t,y,\epsilon_0)\right]-R \left( A(t,y,\epsilon_0),\textbf{u}_i \right)+O(\gamma^{4/3}),
 \vspace{-1mm}
   \end{align}
 where $h_0$ is defined in (\ref{hh0}). 
 
 The rest of the proof focuses on obtaining an expression for $A(t,y,h_0+\epsilon_0)$ in (\ref{eehaa2}) for $y \in \left (y_{i,1}, y_{i,2}\right )$. Consider (\ref{nondim}) for $Z \in(h_0/h,1-h_0/h)$. It can be seen from the selection of $h_0$ in (\ref{hh0}) that the P\'{e}clet number (\ref{pec})  has a large value for a small $\gamma$ and can be written as $P_e(Z)=1/\lambda$ where $\lambda$ at its maximum is of  order $O(\gamma^{2/3})$ at $Z=h_0/h$. The PDE (\ref{nondim}) for $Z \in(h_0/h,1-h_0/h)$ can be expressed as
\begin{align}\label{hey}
\lambda \epsilon^2  \frac{\partial a}{\partial \tau}=\lambda \epsilon^2  \frac{\partial a^2}{\partial Z^2}+\lambda  \frac{\partial a^2}{\partial Y^2} -  \frac{\partial a}{\partial Y},~\lambda=P_e^{-1}(Z),~Z \in(h_0/h,1-h_0/h).
\end{align}
Consider (\ref{hey}) for $t \in(t_{i-1},t_{i})$ where $t_i$ is defined in (\ref{ti}).
 For $t \in (t_{i-1},t_{i})$, the part of the flow chamber, in the range $y \in (y_{i-1,2},y_{i,2})$, is partially occupied with the fluid. Therefore, $\partial a/\partial Y$ in (\ref{hey}) has a significant value such that ${\lambda}^{-1}\frac{\partial a}{\partial Y}$ is of order $O(\lambda^{-1})$ where $\lambda$, in (\ref{hey}), has a small value. Thus, comparing (\ref{Astar}) and (\ref{hey}) concludes  that for $t \in (t_{i-1},t_i)$ the variations of $A(t,y,z)$, for $z \in (h_0,h-h_0)$ and $y \in (y_{i-1,2},y_{i,2})$, is fast whereas $A(t,y,\epsilon_0)$ is slowly varying for $y \in (0,l)$. Therefore, assumption (2) can be justified. 
 By ignoring the diffusion along $y$ and $z$ direction versus the advection term in (\ref{hey}), the dimensionalized form of (\ref{hey}) for $t \in (t_{i-1},t_i)$, $y \in (y_{i-1,2},y_{i,2})$, and $z \in (h_0,h-h_0)$ can be written as
  \begin{align}\label{hey3}
\frac{\partial A}{\partial t}= -v(z)  \frac{\partial A}{\partial y}+O(\gamma),~z 
\in(h_0,h-h_0),~y \in (y_{i-1,2},y_{i,2}),~t \in (t_{i-1},t_i).
\end{align}
The above PDE solution has the following initial condition.
\begin{align}\label{initi}
A(t_{i-1},y,z)=0,~y \in (y_{i-1,2},y_{i,2}),~z \in (h_0,h-h_0).
\end{align}
In order to find the solution of (\ref{hey3}), a boundary condition is required. By applying the divergence  theorem in the outer compartment above biosensor $i-1$, a boundary condition is obtained at $y_{i-1,2}$.

{\bf Determining a boundary condition for (\ref{hey3}) at $y=y_{i-1,2}$:}
Consider the general form of the advection-diffusion equation in the flow chamber (\ref{dimensions}) as
\begin{equation}\label{new_advection}
 \frac{\partial A}{\partial t}=\nabla \cdot \left( \gamma \vec{\nabla}  A - A \textbf{v} \right ),\quad y \in(0,l),\quad z \in (0,h)
 \end{equation}
where $\textbf{v}$ is the velocity vector of the fluid, $\nabla \cdot (\cdot)$ represents the divergence operator, and  $\vec{\nabla}$ denotes gradient. 
Applying the divergence theorem \cite{folland} to (\ref{new_advection}) results in
 \begin{equation}\label{div thrm}
 \int _{\Omega} {\frac{\partial A}{\partial t}\,dy\,dz} =\int _ {\partial \Omega} {\left(\left( \gamma \vec{\nabla}  A - A \textbf{v} \right )\cdot \textbf{n} \right)d{l(y,z)}}
 \end{equation}
 where the left hand side is a surface integral over the bounded domain $\Omega$. The right side of (\ref{div thrm}) is a line integral over the boundary of $\Omega$. The vector $\textbf{n}$ is the outward pointing unit normal field of the boundary $\partial \Omega$. Here, the domain $\Omega$ is selected to be the rectangular region defined as $\Omega=\{(y,z)|z \in (h_0,h), y \in (y_{i-1,1},y_{i-1,2})\}$. The domain $\Omega$ comprises the outer compartment of biosensor $i-1$. Assuming that the multi-compartment model (\ref{eq7}) holds for biosensor $i-1$, the maximum time derivative of the concentration in the outer compartment above biosensor $i-1$ is 
 $O(\gamma)$. Thus, the left hand side of (\ref{div thrm}) is of order $O(\gamma)$. The line integral, on the right hand side, over the upper edge of $\Omega$ at $z=h$ is zero due to insulation  boundary condition in (\ref{inlet}). The middle compartment between sensors $i-1$ and $i-2$ is also in equilibrium with concentration $A_{i-1}+O(\gamma)$. Hence, the line integral over the left edge of $\Omega$ in (\ref{div thrm}), at $y=y_{i-1,1}$, is obtained as
\begin{align}\label{gamb0}
\hspace{-4mm} \int_{h_0}^{h} \left( \gamma \vec{\nabla}  A -A \textbf{v} \right )\cdot(-{\hat{y}} dz)=A_{i-1}\int_{h_0}^h v(z)\,dz + O(\gamma).
\end{align} 
Denoting the concentration in the inner compartment of biosensor $i-1$ by ${a}_{i-1}(t,y)$, the derivative $\left. \frac{\partial A}{\partial z} \right|_{z=h_0}$ can be expressed as $\frac{A_{i-1}-{a}_{i-1}(t,y)}{h_0}+O(h_0)$. Using this expression, the part of the line integral in (\ref{div thrm}), over the bottom edge of $\Omega$ at $z=h_0$, can be written as
\begin{align}\label{gamb2}
\int_{y_{i-1,1}}^{y_{i-1,2}} \left( \gamma \vec{\nabla}  A - A \textbf{v} \right )\cdot {(-\hat{z}dy)} =\frac{\gamma L}{h_0} \left( \bar{a}_{i-1}(t)-A_{i-1}\right) +O(\gamma^{4/3}).
\end{align}
Here, $\bar{a}_{i-1}(t)$ is the average of ${a}_{i-1}(t,y)$ over the length of the biosensor.
By substituting (\ref{gamb0}) and (\ref{gamb2}) in (\ref{div thrm}), the line integral  over the right edge of $\Omega$ at $y=y_{i-1,2}$ is evaluated as
\begin{align}\label{for_eybaba}
\int _ {h_0}^{h} {\left( \gamma \frac{\partial A}{\partial y} - A v(z) \right )\,dz }=-f(t),~y=y_{i-1,2}.
\end{align}
where
\vspace{-4mm}
\begin{align}\label{gt}
f(t)=A_{i-1}(h-h_0)v_1+\frac{\gamma L}{h_0} \left( \bar{a}_{i-1}(t)-A_{i-1}\right)+O(\gamma),\quad v_1=\frac{\int_{h_0}^h v(z)}{h-h0}.
\end{align}
Since the concentration everywhere in $\Omega$ is equal to $A_{i-1}+O(\gamma)$, the derivative in $z$ direction along the left edge of $\Omega$ is of order $O(\gamma)$. Therefore, (\ref{for_eybaba}) results in \begin{align}\label{for_eybaba2}
 \gamma \frac{\partial A}{\partial y} - A v_1  =-\frac{f(t)}{h-h_0},\quad y=y_{i-1,2}, \quad z \in (h_0,h).
\end{align}
where $f(t)$ and $v_1$ are defined in (\ref{gt}). Note that the term $\gamma \frac{\partial A}{ \partial y}$ cannot be considered to be of order $O(\gamma)$ since $\frac{\partial A}{ \partial y}$ at $t=t_{i-1}$ and $y=y_{i-1,2}$ can be of order $O(\gamma^{-1})$ or larger. Equation (\ref{for_eybaba2}) can be considered as a boundary condition for the transport equation (\ref{hey3}).

The solution of the transport equation (\ref{hey3}) with the initial condition (\ref{initi}) and the boundary condition (\ref{for_eybaba2}) for $t \in [t_{i-1},t_i)$ can be found  as \cite{Evans}
\begin{align}\label{sol0}
\hspace{-4mm} A(t,y,z)= h(t-t_{i-1},y-y_{i-1,2},z) + O(\gamma),~ y \in [y_{i-1,2}, y_{i,2}),~z \in(h_0,h-h_0),
\end{align}
where,
\vspace{-4mm}
\begin{align}\nonumber
h(t,y,z)=\left \{ 
\begin{array}{lr}
\frac{v(z)}{\gamma (h-h_0)}\exp{\left[\frac{v_1}{\gamma}\left(y-v(z)t\right)\right]} \int_{0}^{t-\frac{y}{v(z)}} \exp{\left(\frac{v_1v(z)}{\gamma}\tau\right)} f(\tau + t_{i-1} )\,d \tau  \\
+k \exp{\left[\frac{v_1}{\gamma}\left(y-v(z)t\right)\right]},~~~~~~~~~~~~~~y\leq v(z)t \\
0,~~~~~~~~~~~~~~~~~~~~~~~~~~~~~~~~~~~~~\, y>v(z)t \geq 0
\end{array}
\right ..  
\end{align}
Using (\ref{gt}) and integration by parts obtains $h(t,y,z)$ in (\ref{sol0}), for $y \leq tv(z)$, as
\begin{align}\label{h_simple}
& h(t,y,z)=\alpha  A_{i-1} + \frac{\gamma L}{h_0 (h-h_0)v_1}\bar{a}_{i-1}\left(t+t_{i-1}-\frac{y}{v(z)}\right) + \exp{\left[\frac{v_1}{\gamma}\left(y-v(z)t\right)\right]} \\ \nonumber
& \times \left[k-\frac{f(t_{i-1})}{(h-h_0)v_1}-\frac{1}{v_1(h-h_0)} \int_{0}^{t-\frac{y}{v(z)}} {\exp{\left(\frac{v_1 v(z)}{\gamma} \tau \right )} \frac{d f(\tau + t_{i-1})}{d \tau}}\, d \tau\right] +O(\gamma),
\end{align}
where $\alpha$ is defined in (\ref{Ai}). Consider (\ref{h_simple}) for $y \leq tv(z)-\beta \gamma^{1/3}$ where $\beta$ is a positive constant. In (\ref{h_simple}), $\frac{d f(\tau )}{d \tau}$ can be replaced with $\frac{\gamma L}{h_0}\frac{d \bar{a}_{i-1}(\tau)}{d \tau}$ according to (\ref{gt}). According to assumption (1), $\frac{d \bar{a}_{i-1}(t)}{d t} \geq 0$ for $t_{i-1} \leq t <t_{i}$. Thus, the maximum derivative $\frac{d \bar{a}_{i-1}(t)}{d t}$ is obtained as $\frac{\gamma}{h_0^2}A_{i-1}$ from (\ref{eq7}) since $R(A,\textbf{u})\geq0$. Considering assumption (3) and $y \leq tv(z)-\beta \gamma^{1/3}$, the maximum magnitude of the third term on the right hand side of (\ref{h_simple}) is of order $O(\gamma)$. From the above explanation, the maximum change in $\bar{a}_{i-1}$, for $t \in [t_{i-1},t_i)$, is $\frac{\gamma}{h_0^2}A_{i-1}(t_i-t_{i-1})$. Therefore, according to assumption (3), the second term on the right hand side of (\ref{h_simple}) is also of order $O(\gamma)$.
Consequently, it can be expressed that
\begin{align}\nonumber
h(t,y,z)=\alpha A_{i-1} + O(\gamma), \quad y \leq t v(z)-\beta \gamma^{1/3}.
\end{align}
Considering (\ref{sol0}) and the above equation, the concentration $A(t,y,z)$ inside the outer compartment of biosensor $i$ and the previous middle compartment  at $t=t_i$ (when the flow reaches the far end of biosensor $i$) can be obtained as
\begin{align}\label{chaagh3}
A(t_i,y,z)=\alpha A_{i-1} + O(\gamma), \quad y \in [y_{i-1,2}, y_{i,2}),~ z \in(h_0,h-h_0),
\vspace{-3mm}
\end{align}
where  $\alpha$ is defined in (\ref{Ai}).
It is shown at the end of this section that for $t>t_i$ the variations of derivative $\frac{\partial A}{\partial t}$ for $y \in (y_{i-1,2},y_{i,2})$ and $z \in (h_0,h-h_0)$ is of order $O(\gamma)$. Therefore, there exists a time $t^*$ such that for $t \in (t_i,t^*)$, (\ref{Astar}) holds.
The complete proof of (\ref{Astar}) can be found at the end of this section.
The value of $A(t,y,h_0+\epsilon_0)$ in (\ref{eehaa2}), for $t \in (t_i,t^*)$, can be obtained as $A(t,y,h_0+\epsilon_0)=A_i+O(\gamma)$ according to (\ref{Astar}) and (\ref{chaagh3}). Hence, $A(t,y,\epsilon_0)$ in (\ref{eehaa2}) for $t\in (t_i ,t^*)$ is governed by
\begin{align}\label{eehaa3}
h_0 \left . \frac{\partial A}{\partial t} \right |_{z=\epsilon_0}&=\frac{\gamma}{h_0}  \left[ A_i-A(t,y,\epsilon_0)\right]-R \left( A(t,y,\epsilon_0),\textbf{u}_i (t,y)\right)+O(\gamma^{4/3}),\\ \nonumber
  y & \in (y_{i,1},y_{i,2}),\quad t \in (t_i,t^*), \quad A(t_i,y,\epsilon_0)=0.
\end{align}
where $A_i$ is obtained in (\ref{Ai}). 
Considering assumption (2) and  $\epsilon_0=O(\gamma^2)$, (\ref{u_n}) can be expressed as
\begin{align}\nonumber
\frac{d\textbf{u}_i(t,y)}{dt}=G(\textbf{u}_i(t,y),{A(t,y,\epsilon_0)}+O(\gamma^2)),~   t \in (t_i, t^*),~y \in \left[ y_{i,1}, y_{i,2}\right], ~ {\textbf{u}}_i \left ( t_i ,y\right )=u_0.
\vspace{-2mm}
\end{align}
For an arbitrary value $y_1$ for $y$, (\ref{eehaa3}) together with the above equation will represent an ODE system whose solution includes expressions in terms of $t$ and independent of $y_1$ for $A(t,y_1,\epsilon_0)$ and $\textbf{u}_i(t,y_1)$. In fact, $A(t,y,\epsilon_0)$ and $\textbf{u}_i(t,y)$ at every point $y$ has the same dynamics in time. In order to derive (\ref{eq7}), $A(t,y,\epsilon_0)$ and $\textbf{u}_i(t,y)$ in (\ref{eehaa3}) and the above equation are replaced with their average values in the range $(y_{i,1},y_{i,2})$. The spatial average of $A(t,y,\epsilon_0)$ in $(y_{i,1},y_{i,2})$ is denoted by $\bar{a}_i(t)$.

To complete the proof, it needs to be shown that the multi-compartment model (\ref{eq7}) holds for the first biosensor which is straightforward. The  concentration in the outer compartment of the first biosensor achieves equilibrium fast. For $t \geq t_1$, it is equal to the concentration $A^*$ at the inlet of the flow chamber.

{\bf Proof of (\ref{Astar}):}  From assumption (1), $A(t,y,z)$ is an increasing concave function in $t$ for $t>t_i$ and $y \in (y_{i-1,2},y_{i,2})$. Thus $\frac{\partial A}{\partial t}\geq 0$ is decreasing in $t$ for $t>t_i$ and $y \in (y_{i-1,2},y_{i,2})$. By showing that
\begin{align}\label{bahbah}
\frac{\partial A(t_i,y,z)}{\partial t}=O(\gamma),~y \in (y_{i-1,2},y_{i,2}- \gamma^{1/3}],~ z \in(h_0,h-h_0),
\end{align}
it can be concluded that $\frac{\partial A}{\partial t}=O(\gamma)$ for all $t \geq t_i$ and thus (\ref{Astar}) is obtained.
For the derivation of (\ref{bahbah}), the value of $\frac{\partial}{\partial y} A(t,y,z)$ for $y \in (y_{i-1,2},y_{i,2})$ at $t=t_i$ can be obtained from (\ref{sol0}), (\ref{h_simple}), and (\ref{gt}) as
\begin{align}\label{khosham}
&v(z) \frac{\partial A(t_i,y,z)}{\partial y}=  -\frac{\gamma L}{h_0 v_1 (h-h_0)} \frac{d \bar{a}_{i-1}}{dt}(t_i - \frac{\Delta y}{v(z)}) \\ \nonumber
&+ \exp \left ({\frac{v_1}{\gamma} (\Delta y -v(z) \Delta t)} \right ) \left [ \frac{v(z)v_1 k}{\gamma}-\frac{v(z)f(t_{i-1})}{\gamma(h-h_0)}  + \frac{\gamma L}{h_0 v_1(h-h_0)} \frac{d \bar{a}_{i-1}(t_{i-1})}{dt} \right ]  \\ \nonumber
&+\frac{\gamma L}{h_0 v_1(h-h_0) }  \exp \left ({\frac{v_1}{\gamma} (\Delta y -v(z) \Delta t)} \right ) \int_0^{\Delta t 
- \frac{\Delta y}{v(z)}} \exp {\frac{v_1 v(z) \tau}{\gamma}} \frac{d^2 \bar{a}_{i-1}}{d t^2}(\tau + t_{i-1})\, d \tau
\end{align}
where $\Delta y = y - y_{i-1,2}$ and $\Delta t =t_i -t_{i-1}$.  The first term on the right hand side of (\ref{khosham}) is of order $O(\gamma)$. The second term  converges to zero faster that $O(\gamma)$ for $y \leq y_{i,2} - \beta \gamma^{1/3}$. Here, $\beta$ is a positive constant. Assuming that $\bar{a}_{i-1}(t)$ is concave and $\frac{d^2 \bar{a}_{i-1}}{d t^2} \leq 0$, the maximum magnitude of the integral on the right hand side of (\ref{khosham}) is obtained by setting the exponential term in the integral equal to $\exp {\frac{v_1v(z)}{\gamma} \left( \Delta t -\Delta y /v(z) \right)}$. Therefore, the maximum magnitude of the last term on the right hand side of (\ref{khosham}) is  of order $O(\gamma)$. Thus, 
\begin{align}\nonumber
v(z) \frac{\partial A(t_i,y,z)}{\partial y}=O(\gamma),~y \in (y_{i-1,2},y_{i,2}- \gamma^{1/3}],~ z \in(h_0,h-h_0).
\end{align}
By substituting the above equation into (\ref{e1}), (\ref{bahbah}) is obtained.

\section{The proof of Theorem~\ref{consist} on the asymptotic properties of the estimator}
\label{consistproof}
For the proof of part (1), it is required to show that $g_i(A_1,t)$ in (\ref{eq70}) is continuous in $A_1$. 
In the multi-compartment model (\ref{eq7}), $\bar{\textbf{u}}_i$ and $\bar{a}_i$  are continuous in $A_i$ since $R$ and $G$ are continuous functions in $\bar{\textbf{u}}_i$, $\bar{a}_i$, and $A_i$ and Lipschitz in $\bar{\textbf{u}}_i$ and $\bar{a}_i$ with Lipschitz constant independent of $A_i$ \cite{coddington}. From (\ref{response}), $g_i(A_1,t)=F(\bar{\textbf{u}}_i)$. Since $F$ models an electrical response as a measurable physical quantity, it is continuous in $\bar{\textbf{u}}_i$. Thus, $g_i(A_1,t)$ is continuous in $\bar{\textbf{u}}_i$. In our case study in Sec.\ref{result}, $F$ is equal to one of the elements of $\bar{\textbf{u}}_i$. On the other hand, from the recursions in (\ref{Ai}), $A_i$ is continuous in $A_1$. The continuity of $g_i(A_1,t)$ in $\bar{\textbf{u}}_i$, $\bar{\textbf{u}}_i$ in $A_i$, and $A_i$ in $A_1$ concludes that $g_i(A_1,t)$ is also continuous in  $A_1$. Relating the current estimation problem  to the one stated in Theorem~\ref{consis_in_paper}, $\textbf{f}_t$ in the theorem corresponds to the vector $\left[g_1(A_1,t^{i,k}),g_2(A_1,t^{i,k}),\ldots ,g_N(A_1,t^{i,k})\right]^T$. The noise vectors $\left[n_1^k,n_2^k,\ldots ,n_N^k\right]^T$ for $k=1,\ldots ,S$ are independent identically distributed with covariance $\sigma^2 I$. Dealing with condition (b) in Theorem~\ref{consis_in_paper}, it can be shown that the tail cross product of $g_i(A_1,t^{i,k})$ exists. It is required to show that for each $A_1,A_2>0$, $S^{-1} \sum_{k=1}^S g_i(A_1,t^{i,k})g_i(A_2,t^{i,k})$ converges uniformly to a specific function of $A_1$ and $A_2$ when $S,t^{i,S} \to \infty$. When $t^{i,S} \to \infty$, the concentration of chemical species on each biosensor eventually achieves equilibrium. So does the response $g_i$ which is a function of the concentration of chemical species. Thus, $g_i(A,t^{i,S}) \to g_{e,i}(A)$  where $g_{e,i}(A)$ is the steady state value of the response. Assuming that the limit function $g_{e,i}(A)$ is continuous, it can be concluded that $g_i(A,t^{i,S}) \to g_{e,i}(A)$ uniformly as $S \to \infty$. Therefore, $S^{-1} \sum_{k=1}^S g_i(A_1,t^{i,k})g_i(A_2,t^{i,k})$ converges uniformly to $g_{e,i}(A_1)g_{e,i}(A_2)$. In order to satisfy assumption (c), the existence of corresponding tail cross products can be similarly proved.
 
 There is a one-to-one correspondence between the value of concentration $A_1$ and  the concentration of chemical species at equilibrium. Since there is also a one-to-one correspondence between the response and the concentration of chemical species on the biosensor, $g_i(A_1,t^{i,k})$ is a bijection function of $A_1$. Therefore, the only value of $A_1$ which minimizes
$Q(A_1)=\lim_{S \to \infty} S^{-1} \sum_{k=1}^S  \sum_{i=1}^N \left ( g_i(A^*,t^{i,k})-g_i(A_1,t^{i,k}) \right)^2$
is $A_1=A^*$. Thus, condition (b) in Theorem~\ref{consis_in_paper} is satisfied and $\hat{A_1}$ is strongly consistent.
 
In our case, $a(\theta_0)$ in Theorem~\ref{asym} corresponds to 
\begin{align}\label{gamma2} 
\textstyle
 \Gamma= \lim_{S \to \infty } S^{-1} \sum_{k=1}^S\sum_{i=1}^N \left (\partial g_i(A^*,t^{i,k}) / \partial A_1 \right )^2,
 \end{align}
where  $\partial g_i(A^*,t^{i,k}) / \partial A_1$ is the value  of $\partial g_i(A_1,t^{i,k}) / \partial A_1$ at $A_1=A^*$. We assume that all biosensors have identical parameters and binding properties. 
 Thus, according to (\ref{eq7}), there is a common functional relation between the response of each biosensor and the concentration in its outer compartment expressed as
 \vspace{-1mm}
 \begin{equation}\label{samefun}
g_i(A_1,t)=g(A_i,t-t_i) \quad \text{for} \quad i=1,\ldots ,N
\end{equation}
\vspace{-1mm}
Here, $g(A_i,t)$ denotes the response of any biosensor when the concentration in its outer compartment is $A_i$. The response $g(A,t)$ commences at $t=0$. Recall the time shift $t_i$ in (\ref{samefun}), is the response delay of biosensor $i$.   According to the recursion (\ref{Ai}) in Theorem~\ref{multiple-comp}, (\ref{samefun}) can be expressed as $g_i(A_1,t)=g(\alpha ^{i-1}A_1,t-t_i)$.  Therefore, $\Gamma$ in (\ref{gamma2}) can be expressed by (\ref{gamma}). 
 Regarding Theorem~\ref{asym}, the estimation error $\hat{A}_1-A^*$ is asymptotically normal.
%

\section{Proof of Theorem~\ref{pos}}
\label{positivity}
In order to prove the positivity, we rely on the fact that the solution of the spatially discretized problem converges to the solution of the original PDE problem when the spatial discretization step is sufficiently small. Thus we can conclude the positivity of the solution of PDE by proving that the solution of the discretized system is positive. We use the following theorem to prove the positivity of the discretized system:
\begin{theorem} 
\label{thrpos}
Suppose that $F(t,v)$ is continuous and satisfies a Lipschitz condition with respect to $v$. Then the system $w^{\prime}(t)=F(t,w(t))$ is positive  iff for any vector $v  \in \mathbb{R} ^m$ and all $i=1,...,m$ and $t \geq 0$,
\begin{equation} 
\label{posthrm}
v \geq 0, \, v_i=0 \, \Longrightarrow F_i(t,v)\geq 0
\end{equation}
\end{theorem}
By positivity of the system $w^{\prime}(t)=F(t,w(t))$ we mean that for any initial value $w(0) \geq 0$ the solution is $w(t) \geq 0$ for all $t\geq 0$.
\begin{figure}[t]
 \centering
  \includegraphics[width=.5\textwidth]{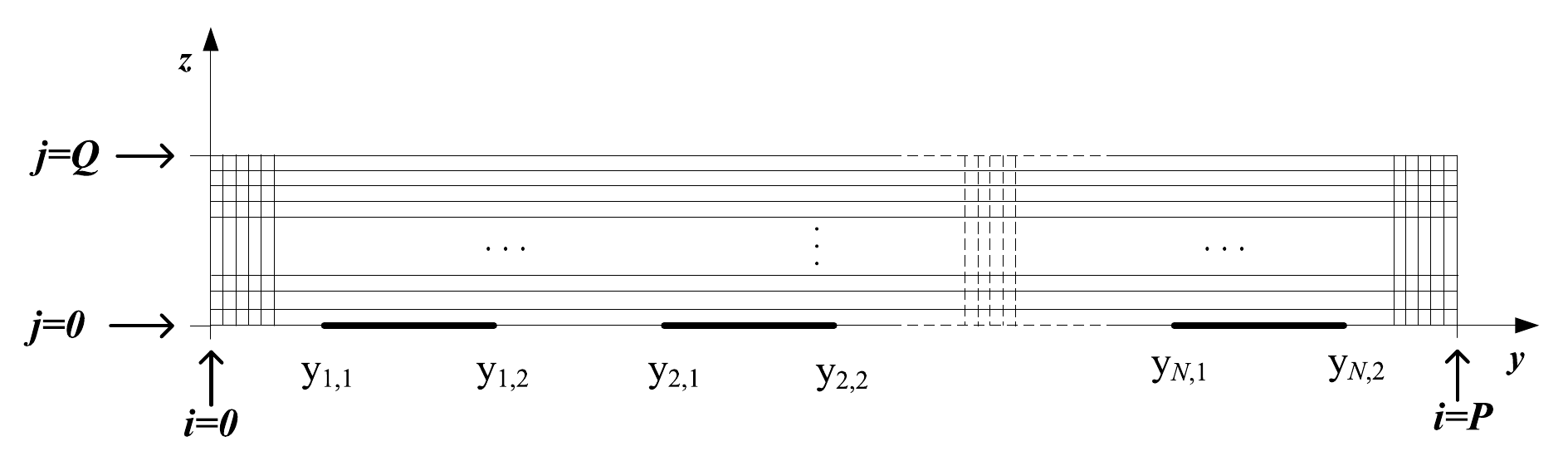}
   \caption{\label{discrete} \small{Discretizing the flow chamber along the length and height of chamber. The length and width are respectively discretized to $P+1$ and $Q+1$ steps }}
   \vspace{-2mm}
 \end{figure}
The requirement of a Lipschitz condition can be slightly relaxed. It is sufficient that the initial value problem has a unique solution for any $w(0) \geq 0$ \cite{Runge}. The proof of this theorem can be found in \cite{numerical}. 

Consider the General PDE model described in Sec. \ref{sec:II}. It can be proved that the solution of this model exists and is unique \cite{Evans}. We discretize the flow chamber in $y$ and $z$ directions to $P+1$ and $Q+1$ steps respectively as shown in Fig. \ref{discrete}. The inlet locates at $y=0$ which is equivalent in the discretized system to $i=0$ and the outlet is at $i=P$. The array of sensors is located at $j=0$ and the upper face of the chamber is at $j=Q$. By applying first and second order approximations, the discretized form of the advection-diffusion equation inside the chamber can be written as
\begin{align} \label{Aij}
A_{i,j}^{\prime}&=\gamma \frac{A_{i-1,j}-2A_{i,j}+A_{i+1,j}}{\Delta y^2}-v(j)\frac{A_{i+1,j}-A_{i,j}}{\Delta y} 
+\gamma \frac{A_{i,j-1}-2A_{i,j}+A_{i,j+1}}{\Delta z^2} \\ \nonumber
&\text{for}\quad i=1,...,P-1,j=1,...,Q-1
\end{align}
Where the discretization steps in $y$ and $z$ direction are respectively referred by $\Delta y>0$ and $\Delta z>0$.
We need to consider $A_{i,j}$ at the boundaries separately through the boundary conditions. The boundary condition at $y=0$ simply converts to
\begin{equation} \label{boundaryaty0}
 A_{0,j}=A_1, \quad j=0,1,...,Q.
\end{equation}
The discretized form of boundary condition at $y=l$ can be written as
\begin{equation}\label{boundaryatyL}
\frac{A_{P,j}-A_{P-1,j}}{\Delta y}=0,\quad j=0,1,...,Q.
\end{equation}
similarly we have the insulation condition on the upper side of chamber;
\begin{equation}\label{boundaryatzh}
\frac{A_{i,Q}-A_{i,Q-1}}{\Delta z}=0\quad i=0,1,...,P.
\end{equation}
At $z=0$, at points which are not located on sensors, the boundary condition can be written as
\begin{equation} \label{boundaryatsensor1}
\frac{A_{i,1}-A_{i,0}}{\Delta z}=0\ \text{for}  \, i \in \{0,1,...,P\} - {\cup}_{k=1}^N V_k .
\end{equation}
where $V_k$ is defined as the set of indices of the discretized points which are located on the $k$-th sensor;
\begin{equation}\label{Vk}
V_k=\{ i \in \{0,1,...,P\} \mid i \Delta y \in [y_{k,1},y_{k,2}] \}
\end{equation}
The boundary condition on sensors, according to (\ref{boundary_on_ICS}), can be written as
\begin{equation}\label{boundaryatsensor2}
\gamma \frac{A_{i,1}-A_{i,0}}{\Delta z}=A_{i,0} q^{t}\textbf{u}_k^i-p^{t}\textbf{u}_k^i, \quad i \in V_k, \quad \text{for} \quad i=1,...,N.
\end{equation}
where $\textbf{u}_k^i=u_k(t,i \Delta y)$ for $i \in V_k$. 
The system of ODEs in (\ref{new odes}) remains unchanged for the discretized domain. At each point, the time derivative of the chemical concentration vector on $k$-th sensor can be written as
\begin{align}\label{uki}
{\textbf{u}_k^i}^{\prime}&=M f(\textbf{u}_k^i,A_{i,0})\, \quad  i \in V_k, \quad \text{for} \quad i=1,...,N. 
\end{align}
For proving the positivity (short for 'non-negativity preserving') of the solution, we need to show that $A_{i,j} \geq 0$ for $i=0,...,P$ and $j=0,...,Q$ and $\textbf{u}_k^i \geq$ for $i \in V_k$ where $k=1,2,...,N$. Using Theorem~\ref{thrpos}, we select $v$ to be a vector containing $A_{i,j}$ for $i=1,...,P-1$ and $j=1,...,Q-1$ and $\textbf{u}_k^i$ for $i \in V_k$ and $k=1,...,N$ as its elements. $v$ can be represented as
\begin{equation}\label{v}
v=\ [ \ [ A_  { i,j  } \ ]_{i \in \{1,...,P-1\} , j \in \{ 1,...,Q-1 \}} , \ [ \textbf{u}_k^i\ ] _{i \in V_k, k \in \{1,...,N\}}  \ ]
\end{equation}

Apparently,  $F(t,v)$ contains their derivatives given by (\ref{Aij}) and (\ref{uki}). Considering Theorem \ref{thrpos}, it is required to show that the statement in (\ref{posthrm}) holds for this case. To this end, we split the elements of $v$ into several sets and investigate (\ref{posthrm}) for each one separately. First set is denoted by $v_1$ and contains the elements $A_{i,j}$ for $i=2,...,P-2$ and $j=2,...,Q-2$ such that
$ v_1= \{ A_ { i,j  } |i \in \{2,...,P-2\} , j \in \{ 2,...,Q-2 \} \}$. We should show for $A_{i,j} \in v_1$ that $A_{i,j}^{\prime}$ is not negative when $A_{i,j}=0$ and $v\geq 0$. Considering these assumptions and according to (\ref{Aij}), $A_{i,j}^{\prime}$ is obtained as
\begin{align} \label{Aij0}
&A_{i,j}^{\prime}=\gamma \frac{A_{i-1,j}+A_{i+1,j}}{\Delta y^2}-v(j)\frac{A_{i+1,j}}{\Delta y}+\gamma \frac{A_{i,j-1}+A_{i,j+1}}{\Delta z^2} \\ \nonumber
&i=2,...,P-2,j=2,...,Q-2
\end{align}
We have $i \in \{ 2,...,P-2\} $ and $j \in \{ 2,...,Q-2 \}$. Therefore, $A_{i-1,j}$, $A_{i+1,j}$, $A_{i,j-1}$, and $A_{i,j+1}$ belong to the vector $v$ and are positive according to the assumption in (\ref{posthrm}). It is easy to show that for $\Delta y < \gamma / \bar{v}$, $A_{i,j}^{\prime}$ in (\ref{Aij0}) is non-negative.
The second part of $v$ is denoted by $v_2$ and comprises the values of $A_{1,j}$ for $j=2,...,Q-2$;
\begin{equation}\label{v2}
v_2= \{ A_{ 1,j  } | j \in \{ 2,...,Q-2 \}\}
\end{equation}
According to (\ref{Aij}), $A_{1,j}^{\prime}$ is written in terms of $A_{0,j}$, $A_{2,j}$, $A_{1,j-1}$, and $A_{1,j+1}$. For the values of $j \in \{ 2,...,Q-2\}$, all these elements, except $A_{0,j}$, belong to the vector $v$ which is assumed to be non-negative. $A_{0,j}$ does not belong to $v$ but it is non-negative due to the boundary condition described in (\ref{boundaryaty0}). Similarly, it can thus be shown that for $\Delta y < \gamma /\bar{v}$, the time derivative of the second part of $v$ which is given in (\ref{v2}) is non-negative. 
The third part of $v$ contains $A_{P-1,j}$ for $j=2,...,Q-2$ and is denoted by $v_3=\{ A_{P-1,j}| j\in\{2,...,Q-2\}\}$.
In a similar way, it can be shown that elements of $v_3$ are non-negative. To prove this, we need to show that $A_{P,j}$ which does not belong to $v$ is non-negative. According to the boundary condition in (\ref{boundaryatyL}), it is obvious that
\begin{equation}\label{boundaryatyL2}
A_{P,j}=A_{P-1,j}
\end{equation}
On the other hand, the assumption in (\ref{posthrm}) implies that $A_{P-1,j}$ should be set to zero. Hence, $A_{P,j}=A_{P-1,j}=0$ according to (\ref{boundaryatyL2}). Consequently, $A_{P-1,j}^{\prime}$ can be written as
\begin{equation}\label{AN}
A_{P-1,j}^{\prime}=\gamma \frac{A_{P-2,j}}{\Delta y^2}+\gamma \frac{A_{P-1,j-1}+A_{P-1,j+1}}{\Delta z^2} 
\end{equation}
Since all the terms on the right-hand side of (\ref{AN}) is non-negative, $A_{P-1,j}^{\prime}$ for $j=2,...,Q-2$ is also non-negative.
The fourth part of $v$ is considered as $v_4=\{ A_{i,1}| i\in\{2,...,P-2\} \}$.
Its time derivative can be written as
\begin{align} \label{Ai1}
&A_{i,1}^{\prime}=\gamma \frac{A_{i-1,1}+A_{i+1,1}}{\Delta y^2}-v(1)\frac{A_{i+1,1}}{\Delta y}+\gamma \frac{A_{i,0}+A_{i,2}}{\Delta z^2} \\ \nonumber
&\text{for}\quad i=2,...,P-2
\end{align}
In (\ref{Ai1}), $A_{i,0}$ is the only element that is not included in $v$ and its positivity should be investigated. According to the boundary conditions in (\ref{boundaryatsensor1}) and (\ref{boundaryatsensor2}), $A_{i,0}$ can be written as
\begin{equation} \label{Ai0}
 A_{i,0}= \left \{
\begin{array}{cc}
\frac{A_{i,1}+\frac{\Delta z}{\gamma} p^t \textbf{u}_k^i} {1+\frac{\Delta z}{\gamma} q^t\textbf{u}_k^i}  & i \in V_k \quad \text{for} \quad k=1,...,N \\
A_{i,1} & i \in \{0,1,...,P\} - {\cup}_{k=1}^N V_k
\end{array} \right .
\end{equation}
where $V_k$ is defined in (\ref{Vk}).According to our assumption, $A_{i,1}=0$. Besides, considering the assumption that $v\geq 0$, we have $\textbf{u}_k^i \geq 0$. Therefore, (\ref{Ai0}) shows that $A_{i,0}$ is not negative. It can be shown that for the same small value of $\Delta y < \gamma / \bar{v}$, $A_{i,1}^{\prime} \geq 0$ for $i=2,...,P-2$.

The fifth part of $v$ to investigate is described as
$v_5=\{ A_{i,Q-1}| i \in \{2,...,P-2\} \}$.
According to (\ref{posthrm}) and (\ref{Aij}), We obtain
\begin{align}\label{AiM_1}
&A_{i,Q-1}^{\prime}=\gamma \frac{A_{i-1,Q-1}+A_{i+1,Q-1}}{\Delta y^2}-v(Q-1) \frac{A_{i+1,Q-1}}{\Delta y} \\ \nonumber
&+ \gamma \frac{A_{i,Q-2}+A_{i,Q}}{\Delta z ^2} \quad \text{for} \quad i=2,...,P-2
\end{align}
In (\ref{AiM_1}), the only term that does not belong to the vector $v$ is $A_{i,Q}$. According to the insulation boundary condition in (\ref{boundaryatzh}), $A_{i,Q}$ is equal to $A_{i,Q-1}$ which is assumed to be zero. Therefore, (\ref{AiM_1}) is reduced to
\begin{align}\label{AiM_1_2}
&A_{i,Q-1}^{\prime}=\gamma \frac{A_{i-1,Q-1}+A_{i+1,Q-1}}{\Delta y^2}-v(Q-1) \frac{A_{i+1,Q-1}}{\Delta y} \\ \nonumber
&+ \gamma \frac{A_{i,Q-2}}{\Delta z ^2} \quad \text{for} \quad i=2,...,P-2
\end{align}
and according to the previous explanations, it is easy to realize from (\ref{AiM_1_2}) that $A_{i,Q-1}^{\prime}$ is not negative for $\Delta y < \gamma / \bar{v}$.
 The remaining elements of $v$ for checking the statement in (\ref{posthrm}) about, are $A_{1,1}$, $A_{1,Q-1}$, $A_{P-1,1}$, $A_{P-1,Q-1}$, and $\textbf{u}_k^i$ for $i \in V_k$ and $k=1,...,N$.
 
 For $A_{1,1}$, it can be written that
\begin{align} \label{A11}
&A_{1,1}^{\prime}= \gamma \frac{A_{0,1}+A_{2,1}}{\Delta y^2} - v(1) \frac{A_{2,1}}{\Delta y} + \gamma \frac{A_{1,0}+A_{1,2}}{\Delta z^2}
\end{align}
$A_{2,1}$ and $A_{1,2}$ belong to $v$ and are therefore non-negative. Regarding (\ref{boundaryaty0}), $A_{0,1}=A_1 \geq 0$. $A_{1,0}$ is achieved by either (\ref{boundaryatsensor1}) or (\ref{boundaryatsensor2}) as
\begin{equation} \label{A10_1}
A_{1,0}= \left \{
\begin{array}{cc}
\frac{\frac{\Delta z}{\gamma} p^t \textbf{u}_k^1} {1+\frac{\Delta z}{\gamma} q^t\textbf{u}_k^1}  & 1 \in V_k \quad \text{for} \quad k=1,...,N \\
A_{1,1} & 1 \notin {\cup}_{k=1}^N V_k
\end{array} \right .
\end{equation}
It is concluded from (\ref{A10_1}) that $A_{1,0}$ is not negative. 
Consequently, according to (\ref{A11}), $A_{1,1}^{\prime} \geq0$ for $\Delta y < \gamma / \bar{v}$.

Equation (\ref{Aij}) for $A_{1,Q-1}$ is written as
\begin{align}\label{A_1QQ}
A_{1,Q-1}^{\prime}= \gamma \frac{A_{0,Q-1}+A_{2,Q-1}}{\Delta y^2} -v(Q-1) \frac{A_{2,Q-1}}{\Delta y} 
+ \gamma \frac{A_{1,Q-2}+A_{1,Q}}{\Delta z ^2}
\end{align}
All terms on the right side of (\ref{A_1QQ}) other than $A_{0,Q-1}$ and $A_{1,Q}$ are included in $v$ and are non-negative. According to the boundary condition at inlet in  (\ref{boundaryaty0}), $A_{0,Q-1}=A_1 \geq 0$. $A_{1,Q}$ is equal to $A_{1,Q-1}$ according to (\ref{boundaryatzh}) which is assumed to be zero. Consequently, $A_{1,Q-1}^{\prime}$ can be proved to be non-negative for the same small value of $\Delta y$.

Considering (\ref{Aij}) for $i=P-1$ and $j=1$ and setting $A_{P-1,1}$ to zero, we can state that
\begin{align}\label{900}
&A_{P-1,1}^{\prime} = \gamma \frac{A_{P-2,1}+A_{P,1}}{\Delta y ^2} -v(1) \frac{A_{P,1}}{\Delta y} \\ \nonumber
& + \gamma \frac{A_{P-1,0}+A_{P-1,2}}{\Delta z^2}
\end{align}
$A_{P,1}$ can be shown to be non-negative according to (\ref{boundaryatyL}). According to (\ref{Ai0}) and the assumption that $A_{P-1,1}=0$, $A_{P-1,0}$ is greater than or equal to zero. The other elements are all included in $v$ and are positive according to assumption in  (\ref{posthrm}). Therefore, it can be easily seen from (\ref{900}) that $A_{P-1,1}^{\prime} \geq 0$ for the same small value of $\Delta y$.

For $A_{P-1,Q-1}$, we rewrite (\ref{Aij}) for $i=P-1$ and $j=Q-1$ as
\begin{align}\label{eq99}
A_{P-1,Q-1}^{\prime}= \gamma \frac{A_{P-2,Q-1}+A_{P,Q-1}}{\Delta y^2}-v(Q-1) \frac{A_{P,Q-1}}{\Delta y} 
+ \gamma \frac{A_{P-1,Q-2}+A_{P-1,Q}}{\Delta z ^2}
\end{align}
According to boundary conditions in (\ref{boundaryatyL}) and (\ref{boundaryatzh}), $A_{P,Q-1}=A_{P-1,Q-1}=0$ and $A_{P-1,Q}=A_{P-1,Q-1}=0$. Therefore, (\ref{eq99}) can be rewritten as 
\begin{align}\nonumber
A_{P-1,Q-1}^{\prime}= \gamma \frac{A_{P-2,Q-1}}{\Delta y^2}+ \gamma \frac{A_{P-1,Q-2}}{\Delta z ^2}\nonumber
\end{align}
which shows that $A_{P-1,Q-1}^{\prime}$ is positive since $A_{P-2,Q-1}$ and $A_{P-1,Q-2}$ belong to $v$ and are positive.
Reminding the definition of $v$ in (\ref{v}), we should check the validity of (\ref{posthrm}) for the remaining elements $\textbf{u}_k^i$ when $i \in V_k$ for $k=1,...,N$. $\textbf{u}_k^i$ is the concentration vector at point $i$ on $k$-th sensor. Since the proof of (\ref{posthrm}) for each element of this vector is the same and independent of the others, we refer to each $\textbf{u}_k^i$ by $\textbf{u}$ for simplicity. $\textbf{u}$ has eight elements and can be presented as $\textbf{u}=[u_1~u_2~...u_8]$. To complete the validity of (\ref{posthrm}), we should prove that when $v\geq0$ and $u_j=0$, the time derivative ${u_j}^{\prime}$ is not negative for $j=1,...,8$. According to (\ref{uki}), ${u_j}^{\prime}$ can be written as
\begin{equation}\label{u_j}
{u_j}^{\prime}=M_j f(\textbf{u},A_{i,0})
\end{equation}
 where $M_j$ is the $j$-th row of $M$. According to (\ref{Ai0}), $A_{i,0}$ on the sensor can be written as
\begin{equation}\nonumber
 A_{i,0}=\frac{A_{i,1}+\frac{\Delta z}{\gamma} p^t \textbf{u}} {1+\frac{\Delta z}{\gamma} q^t\textbf{u}}
 \end{equation}
 We have already proved that $A_{i,1}$ is positive for $i=1,...,P-1$. It is easy to show that $A_{0,1}$ and $A_{P,1}$ are positive. Consequently, the assumption that $\textbf{u} \geq 0$ results in $A_{i,0} \geq 0$. On the other hand, by considering the elements of $M$ in (\ref{M}) and the definition of $f(\textbf{u},A)$ given in (\ref{r(u,A)}) and (\ref{rate eq}), it can be easily shown that $u_j^{\prime} \geq 0$ when we set $\textbf{u} \geq 0$ and $u_j=0$. From chemical point of view, the only negative terms in (\ref{u_j}) which reduce the amount of $u_j$ correspond to the ones that have $u_j$ as a multiplicative factor. This is due to the fact that the backward reaction only happens when $u_j$ is not zero. However, in our assumptions $u_j$ is zero and the negative terms are therefore omitted from the right side of (\ref{u_j}).
 
 The positivity of the variables $A_{i,j}$ for $i=0$ and $j=1,...,Q-1$ can be easily shown since $A_{0,j}$ is equal to $A_1$ at the inlet and therefore $A_{0,j} \geq 0$  for $j=0,...,Q$. According to the boundary condition in (\ref{boundaryatyL}), $A_{P,j}=A_{P-1,j}$ and since we proved that $A_{P-1,j}$ for $j=1,...,Q-1$ is positive, we can conclude that $A_{P,j} \geq 0$ for $j=1,...,Q-1$. For $j=0$ and $i=1,...,P-1$ we can prove the positivity by referring to (\ref{Ai0}) because we have proved that $A_{i,1}\geq 0$ for $i=1,...,P-1$ and $\textbf{u}_k^i \geq 0$ for $i \in V_k$ and $k=1,...,N$. For $j=0$ and $i=P$, we can use the boundary condition in (\ref{boundaryatyL}) and write $A_{P,0}=A_{P-1,0}$ and since we just proved that $A_{P-1,0}\geq 0$, it can be easily seen that $A_{P,0}\geq0$. For $j=Q$ and $i=1,...,P$ we can use the insulation boundary condition given in (\ref{boundaryatzh}) and write $A_{i,Q}=A_{i,Q-1}$. We have already proved that $A_{i,Q-1} \geq 0$ for $i=1,...,P-1$ using Theorem~\ref{thrpos}. The proof of positivity of $A_{P,Q-1}$ is explained earlier in this paragraph. The proof of positivity of the solution of the PDE model of Sec. \ref{sec:II} is now complete.
 
\bibliographystyle{IEEEtran}
\bibliography{ref2}

\begin{thebibliography}{10}
\providecommand{\url}[1]{#1}
\csname url@samestyle\endcsname
\providecommand{\newblock}{\relax}
\providecommand{\bibinfo}[2]{#2}
\providecommand{\BIBentrySTDinterwordspacing}{\spaceskip=0pt\relax}
\providecommand{\BIBentryALTinterwordstretchfactor}{4}
\providecommand{\BIBentryALTinterwordspacing}{\spaceskip=\fontdimen2\font plus
\BIBentryALTinterwordstretchfactor\fontdimen3\font minus
  \fontdimen4\font\relax}
\providecommand{\BIBforeignlanguage}[2]{{%
\expandafter\ifx\csname l@#1\endcsname\relax
\typeout{** WARNING: IEEEtran.bst: No hyphenation pattern has been}%
\typeout{** loaded for the language `#1'. Using the pattern for}%
\typeout{** the default language instead.}%
\else
\language=\csname l@#1\endcsname
\fi
#2}}
\providecommand{\BIBdecl}{\relax}
\BIBdecl

\bibitem{bib1}
B.~Cornell, V.~Braach-Maksvytis, L.~King, P.~Osman, B.~Raguse, L.~Wieczorek,
  and R.~Pace, ``A biosensor that uses ion-channel switches,'' \emph{Nature},
  vol. 387, pp. 580--583, 1997.

\bibitem{banks}
H.~T. Banks and K.~Kunisch, \emph{Estimation Techniques for Distributed
  Parameter Systems}.\hskip 1em plus 0.5em minus 0.4em\relax Birkh\"{a}user,
  1989.

\bibitem{khoshgel}
M.~N. \"{O}zisi¸k and H.~R.~B. Orlande, \emph{Inverse heat transfer:
  fundamentals and applications}.\hskip 1em plus 0.5em minus 0.4em\relax Taylor
  and Francis, 2000.

\bibitem{V}
A.~Shidfar, R.~Zolfaghari, and J.~Damirchi, ``Application of sinc-collocation
  method for solving an inverse problem,'' \emph{Journal of Computational and
  Applied Mathematics}, vol. 233, no.~2, pp. 545--554, 2009.

\bibitem{VI}
D.~V. Antonov, ``Boundary observation of nonstationary neutron transport in a
  plane-parallel medium,'' \emph{Computational Mathematics and Modelling},
  vol.~7, no.~4, pp. 339--344, 1996.

\bibitem{cylinder}
N.~V. Kerov, ``Solution of the two-dimensional inverse heat-conduction problem
  in a cylindrical coordinate system,'' \emph{Computational Mathematics and
  Modelling}, vol.~45, no.~5, pp. 1245--1249, 1983.

\bibitem{16}
T.~Chung, \emph{Computational Fluid Dynamics}.\hskip 1em plus 0.5em minus
  0.4em\relax Cambridge University Press, 2002.

\bibitem{3}
T.~Bader, A.~Wiedemann, K.~Roberts, and U.~D. Hanebeck, ``Model-based motion
  estimation of elastic surfaces for minimally invasive cardiac surgery,'' in
  \emph{Proc. of IEEE International Conference on Robotics and Automation
  (ICRA)}, Rome, April 2007.

\bibitem{17}
G.~E. Karniadakis and S.~Sherwin, \emph{Spectral/hp Element Methods for
  Computational Fluid Dynamics}.\hskip 1em plus 0.5em minus 0.4em\relax Oxford
  University Press, 2005.

\bibitem{Wilson}
J.~D. Wilson and W.~K.~N. Shum, ``A re-examination of the integrated horizontal
  flux method for estimating volatilization from circular plots,''
  \emph{Agricultural and Forest Meteorology}, vol.~57, pp. 281--295, 1992.

\bibitem{Flesch}
T.~K. Flesch, J.~D. Wilson, and E.~Yee, ``Backward-time lagrangian stochastic
  dispersion models and their application to estimate gaseous emissions,''
  \emph{Journal of Applied Meteorology}, vol.~34, pp. 1320--1332, 1995.

\bibitem{Bayesian}
A.~Keats, E.~Yee, and E.-S. Lien, ``Bayesian inference for source determination
  with applications to a complex urban environment,'' \emph{Atmospheric
  environment}, vol.~41, no.~3, pp. 465--479, January 2007.

\bibitem{cdc}
S.~Monfared, V.~Krishnamurthy, and B.~Cornell, ``Chemical kinetics and mass
  transport in an ion channel based biosensor,'' in \emph{Proc. of IEEE
  Conference on Decision and Control}, Atlanta, GA, December 2010, pp.
  4685--4690.

\bibitem{Computational}
R.~A. Vijayendran, F.~S. Ligler, and D.~E. Leckband, ``A computational
  reaction-diffusion model for the analysis of transport-limited kinetics,''
  \emph{Analytical Chemistry}, vol.~71, pp. 5405--5412, 1999.

\bibitem{Extending}
D.~G. Myszka, X.~He, M.~Dembo, T.~A. Morton, and B.~Goldstein, ``Extending the
  range of rate constants available from biacore: Interpreting mass
  transport-influenced binding data,'' \emph{Biophysical Journal}, vol.~75, pp.
  583--594, 1998.

\bibitem{hw}
J.~Brody, P.~Yager, R.~Goldstein, and R.~Austin, ``Biotechnology at low
  reynolds numbers,'' \emph{Biophysical Journal}, vol.~71, pp. 3430--–3441,
  1996.

\bibitem{bib6}
V.~Krishnamurthy, S.~Monfared, and B.~Cornell, ``Ion-channel biosensors part
  {II}: Dynamic modeling, analysis and statistical signal processing,''
  \emph{IEEE Trans. on Nanotechnology}, vol.~9, no.~3, pp. 33--321, 2010.

\bibitem{singular}
P.~Kokotovi\'{c}, H.~K. Khalil, and J.~O'Reilly, \emph{Singular perturbation
  methods in Control: Analysis and design}.\hskip 1em plus 0.5em minus
  0.4em\relax Siam, 1999.

\bibitem{asymptotic}
R.~I. Jennrich, ``Asymptotic properties of non-linear least squares
  estimators,'' \emph{The Annals of Mathematical Statistics}, vol.~40, no.~2,
  pp. 633--643, 1969.

\bibitem{bib5}
V.~Krishnamurthy, S.~Monfared, and B.~Cornell, ``Ion-channel biosensors part
  {I}: Construction, operation and clinical studies,'' \emph{IEEE Trans. on
  Nanotechnology}, vol.~9, no.~3, pp. 303--312, 2010.

\bibitem{numerical}
W.~Hundsdorfer and J.~G. Verwer, \emph{Numerical Solution of Time-Dependent
  Advection-Diffusion-Reaction Equations}.\hskip 1em plus 0.5em minus
  0.4em\relax Springer Computational Mathematics, 2003.

\bibitem{folland}
G.~B. Folland, \emph{Introduction to partial differential equations},
  2nd~ed.\hskip 1em plus 0.5em minus 0.4em\relax Princeton University Press,
  1995.

\bibitem{Evans}
L.~C. Evans, \emph{Partial differential equations}, 2nd~ed.\hskip 1em plus
  0.5em minus 0.4em\relax American Mathematical Society, 2010.

\bibitem{coddington}
E.~A. Coddington and N.~Levinson, \emph{Theory of ordinary differential
  equations}.\hskip 1em plus 0.5em minus 0.4em\relax McGraw-Hill, 1955.

\bibitem{Runge}
Z.~Horv\'{a}th, ``Positivity of runge-kutta and diagonally split runge-kutta
  methods,'' \emph{Applied Numerical MathematicsApplied Numerical Mathematics},
  vol.~28, pp. 309--326, 1998.

\end{thebibliography}
\end{document}